\documentclass[10pt]{article}

\usepackage[authoryear,round]{natbib}       % need for \citet*
\usepackage{amsmath}                        % need for subequations
\numberwithin{equation}{section}
\usepackage{graphicx}                       % need for figures
\usepackage{subfigure}
\usepackage{enumitem}                       % need for subfigures
\usepackage{hyperref}
\usepackage{amssymb}                        % gives you \mathbb{} font
\usepackage[mathscr]{eucal}                 % gives you \mathscr font
\usepackage{pstricks}
\usepackage{rotating}
\usepackage{lscape}
\usepackage[paperwidth=8.5in,paperheight=11in,top=1.25in, bottom=1.25in, left=1.00in, right=1.00in]{geometry}
\usepackage{mathtools}                      % need for `show only references'
\mathtoolsset{showonlyrefs=true}            % only equations which are labeled AND referenced will be numbered.                                                                                                 % IMPORTANT NOTE...must use \eqref{} instead of (\ref{})
\usepackage{fixltx2e,amsmath}               % Supposedly, this allows one to use \eqref{} in \caption{}.
\MakeRobust{\eqref}

\usepackage{mathdots}
\usepackage{amsthm}                         % need for theorem-proof environment
\allowdisplaybreaks                         % allows page breaks for long equations
\theoremstyle{plain}
\newtheorem{theorem}{Theorem}
\numberwithin{theorem}{section}
\newtheorem{lemma}[theorem]{Lemma}          % [theorem] ==> theorems and lemmas will share a counter
\newtheorem{proposition}[theorem]{Proposition}

\theoremstyle{definition}
\newtheorem{definition}[theorem]{Definition}
\newtheorem{example}[theorem]{Example}
\newtheorem{remark}[theorem]{Remark}

\def \a {{\alpha}}
\def \b {{\beta}}
\def \xbar {\bar{x}}

\newcommand{\<}{\langle}
\renewcommand{\>}{\rangle}
\renewcommand{\(}{\left(}
\renewcommand{\)}{\right)}
\renewcommand{\[}{\left[}
\renewcommand{\]}{\right]}

\newcommand\Eb{\mathbb{E}}

\newcommand\Qb{\mathbb{Q}}
\newcommand\Rb{\mathbb{R}}
\newcommand\Ib{\mathbb{I}}

\newcommand\R{\mathbb{R}}

\newcommand\Ac{\mathscr{A}}

\newcommand\Bc{\mathscr{B}}

\newcommand\Ec{\mathscr{E}}
\newcommand\Fc{\mathscr{F}}
\newcommand\Gc{\mathscr{G}}

\newcommand\Lc{\mathscr{L}}
\newcommand\Mc{\mathscr{M}}

\newcommand\Pc{\mathscr{P}}

\renewcommand\phi{\varphi}

\newcommand\Om{\Omega}
\newcommand\sig{\sigma}

\newcommand\gam{\gamma}
\newcommand\Gam{\Gamma}
\newcommand\lam{\lambda}
\newcommand\del{\delta}
\newcommand\Del{\Delta}

\newcommand\xb{\bar{x}}

\newcommand\ub{\bar{u}}
\newcommand\sigb{\bar{\sig}}

\newcommand\Hv{\mathbf{H}}

\newcommand\phih{\hat{\phi}}

\newcommand\Nt{\widetilde{N}}

\renewcommand\d{\partial}

\newcommand\ii{\mathtt{i}}
\newcommand\dd{\mathrm{d}}
\newcommand\ee{\mathrm{e}}

\def \phi {{\varphi}}

\begin{document}

\title{Asymptotics for $d$-dimensional L\'evy-type processes}

\author{
Matthew Lorig
\thanks{ORFE Department, Princeton University, Princeton, USA.
Work partially supported by NSF grant DMS-0739195.}
\and
Stefano Pagliarani
\thanks{Centre de Math\'ematiques Appliqu\'ees, Ecole Polytechnique and CNRS, Route de Saclay, 91128 Palaiseau Cedex, France. Email: {\tt  pagliarani@cmap.polytechnique.fr}. Work partially supported by the Chair {\it Financial Risks} of the {\it Risk Foundation}.}
\and
Andrea Pascucci
\thanks{Dipartimento di Matematica,
Universit\`a di Bologna, Bologna, Italy}
}

\date{This version: \today}

\maketitle

\begin{abstract}
We consider a general $d$-dimensional L\'evy-type process with killing. Combining the classical
Dyson series approach with a novel polynomial expansion of the generator $\Ac(t)$ of the
L\'evy-type process, we derive a family of asymptotic approximations for transition densities and
European-style options prices.  Examples of stochastic volatility models with jumps are provided
in order to illustrate the numerical accuracy of our approach.  The methods described in this
paper extend the results from \cite{pascucci-parametrix}, \cite{pagliarani-heston} and
\cite{lorig-pagliarani-pascucci-4} for Markov diffusions to Markov processes with jumps.
\end{abstract}


%%%%%%%%%%%%%%%%%%%%%%%%%%%%%%%%%%%%%%%
%
%       SECTION: Introduction
%
%%%%%%%%%%%%%%%%%%%%%%%%%%%%%%%%%%%%%%%

\section{Introduction}
%If one models a market as a $d$-dimensional Markov process $X$, then the time evolution of this
%model can be described as the solution of a L\'evy-It\^o stochastic differential equation (SDE).

In a multi-dimensional Markovian setting, the time evolution of a market model is usually
described by the solution $X$ of a L\'evy-It\^o stochastic differential equation (SDE). Such a
model allows for features commonly seen in markets, such as stochastic volatility, jumps, default,
co-integration and correlation.
%Given a Markov model for a market,
Many quantities of interest (e.g., option prices, net present values) can be expressed as
expectations of the form $u(t,x):=\Eb [\varphi(X_T) | X_t=x]$.  Under mild conditions, the
function $u(t,x)$ is the unique classical solution of a partial integro-differential equation
(PIDE).  Unfortunately, closed form and even semi-closed form solutions of these PIDEs are
available only in rare cases. As such, it is important to develop general methods for finding
analytical approximations for the solutions of these PIDEs.
\par
Within the mathematical finance literature, a number of different approaches have been taken for
finding approximate transition densities and option prices for markets described by Markov
processes.  Most of these techniques involve expansions that exploit a small parameter or a
limiting case.  For example, \cite{gobet-smart} develop analytical approximations for models with
local volatility and Gaussian jumps in the small diffusion and small jump frequency/size limits
(see also the recent review paper by \cite{bompisgobet2013}).
\cite{deuschel-friz-jacquier-violante-1} obtain densities for diffusion processes in a small noise
limit. \cite{fpss} find option prices for Black-Scholes-like multiscale models where volatility is
driven by two factors, one running on a fast scale, one running on a slow scale.
\cite{lorig-2,lorig-lozano-1} extend these multiscale techniques to more general diffusions and to
the exponential L\'evy setting. %\ora{As always Stefano and Andrea, please feel free to add
%references to your work or anything else you think is appropriate.}
\par
Recently, \cite{pagliarani2011analytical} introduce a method for finding asymptotic solutions of
parabolic PDEs. The approach, called the \emph{adjoint expansion method}, is extended by
\cite{pascucci,lorig-pagliarani-pascucci-1} to models with jumps and it was further generalized by
\cite{lorig-pagliarani-pascucci-4} to a family of asymptotic expansions for a $d$-dimensional
market described by an It\^o SDE (i.e., a Markov market with no jumps). The method consists of
expanding the pricing PDE in polynomial basis functions, which results in a nested sequence of
Cauchy problems, and deriving analytical solutions for these nested Cauchy problems. In this
paper, we extend the results of
\cite{pascucci,lorig-pagliarani-pascucci-1,lorig-pagliarani-pascucci-4} to the PIDEs that arise
when markets are described by a $d$-dimensional L\'evy-It\^o SDE.  Results presented here also
simplify results from \cite{pascucci,lorig-pagliarani-pascucci-1,lorig-pagliarani-pascucci-4}.
\par
The rest of this paper proceeds as follows.  In Section \ref{sec:model} we present a general
$d$-dimensional market model.  We also describe the kinds of derivative-assets we wish to price,
and we relate the price of such derivative-assets to the solution of a parabolic PIDE.  In Section
\ref{sec:approximating} we introduce the idea of polynomial expansions of the pricing PIDE and in
Section \ref{sec:dyson}, we derive a family of analytical price approximations -- one for each
polynomial expansion of the pricing PIDE.  Lastly, in Section \ref{sec:heston} we provide a
numerical example, illustrating the versatility and accuracy of our methods.

%%%%%%%%%%%%%%%%%%%%%%%%%%%%%%%%%%%%%%%
%
%       SECTION: Model
%
%%%%%%%%%%%%%%%%%%%%%%%%%%%%%%%%%%%%%%%

\section{Market model}
\label{sec:model}
We take, as given, an equivalent martingale measure $\Qb$ defined on a complete filtered probability space $(\Om,\Fc,\{\Fc_t,t\geq0\},\Qb)$. All stochastic processes defined below live on this probability space and all expectations are taken with respect to $\Qb$. The risk-neutral dynamics of our market are described by the following $d$-dimensional Markov L\'evy-type process
\begin{align}
\dd X_t
   &=  \mu(t,X_t) \, \dd t + \sig(t,X_t) \, \dd W_t
                + \int_{\Rb^d} z \, \dd \Nt(t,X_{t-},\dd t,\dd z). \label{eq:dX}
\end{align}
Here $W$ is a standard $m$-dimensional Brownian motion, and $\Nt(\cdot,\cdot,\dd t,\dd z)$, given
by
%_{(t,x)\in \mathbb{R}^+\times\mathbb{R}^d}
%given by
%for every fixed $s \geq 0$, $x \in \Rb^d$ and Borel set $A \in \Bc(\Rb_0^d)$ the process $\Nt(s,x,A)=\{\Nt_t(s,x,A), t \geq 0 \}$ is compensated Poisson process with intensity $\nu(s,x,A)$ where $\nu$ is a L\'evy kernel $\nu : \Rb_+ \times \Rb^d \times \Bc(\Rb_0^d) \to \overline{\Rb}_+$.
\begin{align}
\Nt(t,x,\dd t, \dd z)
    &= N(t,x,\dd t, \dd z) - \nu(t,x,\dd z)\dd t , &
(t,x)
    &\in \mathbb{R}^+\times\mathbb{R}^d ,
\end{align}
is a family of compensated Poisson measures on $\Bc(\mathbb{R})\otimes \Bc(\mathbb{R}^d)$.
The drift vector $\mu$ and volatility matrix $\sig$ map
$\mu : \Rb_+ \times \Rb^d \to \Rb^d$ and
$\sig: \Rb_+ \times \Rb^d \to \Rb^{ d \times m}$, respectively.
We assume the L\'evy kernel $\nu$ satisfies
\begin{align}\label{e12}
 \int_{\Rb^d} \min\{|z|,|z|^2\} \, \bar{\nu}(\dd z) &< \infty, &
 \bar{\nu}(\dd z):=\sup_{(t,x)\in\R_{+}\times\R^{d}} \nu(t,x,\dd z),
%\int_{\Rb_0^d} \min\{|z|,|z|^2\} \, \nu(t,x,\dd z)
%    &< \infty &
%\forall \, t \in \Rb_+, \, \forall \, x \in \Rb^d ,
\end{align}
which is rather standard for L\'evy-type models.
%{\blue which allows us to compensate all jumps in \eqref{eq:dX}. {\bf [Andrea: I don't understand this.]}}
The components of $X$ could represent a number of things such as e.g., economic factors, asset
prices, indices, or functions of these quantities.  We assume a risk-free interest rate of the
form $r(t,X_t)$ where $r : \Rb_+ \times \Rb^d \to \Rb$.  We also introduce a random time $\zeta$,
which is given by
\begin{align}
\zeta
    &=  \inf \big\{ t \geq 0 : \int_0^t \gam(s,X_s) \dd s \geq \Ec \big\} , &
\gam
        &:  \Rb_+ \times \Rb^d \to \Rb_+ ,
\end{align}
with $\Ec$ exponentially distributed and independent of $X$. The random time $\zeta$ could
represent the default time of an asset, the arrival of an economic shock, etc..
\par
Denote by $V$ the no-arbitrage price of a European derivative expiring at time $T$ with payoff
\begin{align}
H(X_T) \, \Ib_{\{\zeta > T\}} + G(X_T) \, \Ib_{\{\zeta \leq T\}}
    &=  \big( H(X_T) - G(X_T) \big) \, \Ib_{\{\zeta > T\}} + G(X_T) .
\end{align}
It is well known (see, for instance, \cite{yorbook}) that
\begin{align}\label{e1}
V_t
    &=  \Eb \[ \ee^{-\int_t^T r(s,X_s)  \dd s} G(X_T) | X_t \] + \\ &\qquad
        \Ib_{\{\zeta>t\}} \Eb \[\ee^{-\int_t^T \( r(s,X_s) + \gam(s,X_s) \) \dd s} \Big(H(X_T)-G(X_T) \Big)  | X_t \] ,  &
t
    &<T . \label{eq:V}
\end{align}
Thus, to value a European-style option, one must compute functions of the form
\begin{align}
u(t,x)
    &:= \Eb \[\ee^{- \int_t^T  \lam(s,X_s)  \dd s}  \phi(X_T) \mid X_t = x \] . \label{eq:v}
\end{align}
Under mild assumptions (see, for instance, \cite{Pascuccibook}), the function $u$, defined by
\eqref{eq:v}, satisfies the Kolmogorov backward equation
\begin{align}
(\d_t + \Ac(t)) u
    &=  0 , &
u(T,x)
    &=  \phi(x), &
x
    &\in \Rb^d,
\label{eq:u.pide}
\end{align}
where the operator $\Ac(t)$ is given explicitly by
\begin{align}
\hspace{-25pt}\Ac(t)
    &=  \int_{\Rb^d} \nu(t,x,\dd z) \( \ee^{\<z,\nabla_x \>} - 1 - \<z , \nabla_x \> \)
            + \frac{1}{2} \sum_{i,j=1}^d  \(\sig \sig^\text{T}\)_{ij}(t,x) \d_{x_i}\d_{x_j}
      + \sum_{i=1}^d \mu_i(t,x) \d_{x_i} -  \lam(t,x)  ,
            \label{eq:A}
\end{align}
with
\begin{align}
\< z , x \>
    &:= \sum_{i=1}^d z_i \, x_i , &
\nabla_x
    &:= (\d_{x_1},\d_{x_2},\cdots,\d_{x_d}) , &
\ee^{\<z,\nabla_x \>} f(x)
    &:= f(x+z) . \label{eq:3defs}
\end{align}
The formal representation of the \emph{shift operator} $\ee^{\<z,\nabla_x \>}$ is motivated by the
fact that its Taylor expansion applied to the function $f(x)$ gives the Taylor expansion of
$f(x+z)$ about the point $x$. As in \cite[Chapter 1]{oksendal2}, we regard the domain of $\Ac(t)$
to be all functions $f:\Rb^d \to \Rb$ such that $\Ac(t)f(x)$ exists and is finite for all $x \in
\Rb^d$.
%\\
%{\blue In regards to the domain of the operator $\Ac(t)$, at this early stage we consider the very general class of functions $f\in C^2(\Rb^d)$ such that $\Ac(t)f(x)$ exists and is finite for all $x$.}

\begin{remark}[Martingale property]
Let us denote by $X^{(i)}$ the $i$th component of the vector $X$ and assume that
\begin{align}
\int_{|z|\ge 1} \ee^{z_{i}} \, \bar{\nu}(\dd z)< \infty,
\end{align}
for some $i\le d$, with $\bar{\nu}$ as in \eqref{e12}. If $S_t:=\Ib_{\{\zeta>t\}}\ee^{X_t^{(i)}}$
is supposed to be a traded asset then, in order for $S$ to be a martingale, the drift $\mu_i$ must
satisfy
\begin{align}
\mu_i(t,x)
    &=  \gam(t,x) - \int_{\Rb^d} \nu(t,x,\dd z) (\ee^{z_i}-1-z_i)
            - \frac{1}{2} \(\sig \sig^\text{T}\)_{ii}(t,x),
\end{align}
To see this, set $H(x)=\ee^{x_{i}}$, $G(x)=0$ and impose $V_t=S_t$ in \eqref{eq:V}.
\end{remark}

%%%%%%%%%%%%%%%%%%%%%%%%%%%%%%%%%%%%%%%
%
%       SECTION: Expansion Basis
%
%%%%%%%%%%%%%%%%%%%%%%%%%%%%%%%%%%%%%%%

\section{General expansion basis}
\label{sec:approximating}
Let us start by rewriting the differential operator \eqref{eq:A} in the more compact form
\begin{align}
\Ac(t)
    &:= \int_{\Rb^d} \nu(t,x,\dd z) \( \ee^{ \<z , \nabla_x \>} - 1 - \<z , \nabla_x \> \)
            + \sum_{|\alpha |\leq 2} a_{\alpha}(t,x) D^{\alpha}_{x} , &
t\in\R,\ x\in
    \mathbb{R}^d , \label{operator_A}
\end{align}
where by standard notations
\begin{align}
\alpha
    &=  (\alpha_1,\cdots,\alpha_d)\in \mathbb{N}^{d}_{0}, &
|\alpha|
    &=  \sum_{i=1}^{d}\alpha_i, &
D_{x}^{\alpha}
    =       \partial^{\alpha_1}_{x_1}\cdots \partial^{\alpha_d}_{x_d}.
\end{align}
In this section we introduce a family of expansion schemes for $\Ac(t)$, which we shall use to construct closed-form approximate solutions (one for each family) of \eqref{eq:u.pide}.
%\\
%\ora{**************BEGIN CHANGES HERE*************}
\begin{definition}
\label{def:A} For $|\a|\le 2$ and $n\le N\in\mathbb{N}_0$, let $a_{\alpha,n}=a_{\alpha,n}(t,x)$
and $\nu_n=\nu_n(t,x,\dd z)$ be such that the following hold:
\begin{enumerate}
%\item[i)] the first $N$ coefficients in the sequence $(a_{\alpha,n}(t,x))_{n \geq 0}$ are continuous functions that depend polynomially on $x$ with $a_{\alpha,0}(t,x) \equiv a_{\alpha,0}(t)$,
\item[(i)] For any $t \in [0,T]$, $a_{\alpha,n}(t,\cdot)$ are polynomial functions with $a_{\alpha,0}(t,x) \equiv a_{\alpha,0}(t)$, and for any $x\in\Rb^{d}$ the functions $a_{\alpha,n}(\cdot,x)$ belong to
$L^{\infty}([0,T])$.
\item[ii)] %for every $(t,x)$, the sequence $(\nu_n(t,x,dz))_{n \geq 0}$ is a sequence of real measures (not necessarily non-negative) on $\mathbb{R}^d$, the first $N$ of which are continuous in $(t,x)$, depend polynomially on $x$
For any $t\in [0,T]$, $x\in\mathbb{R}^d$, we have
 \begin{align}\label{e34}
  \nu_n(t,x,\dd z)
    &=\sum_{|\beta|\le M_{n}}x^{\b}  \nu_{n,\b}(t,\dd z), & M_n &\in \mathbb{N}_{0} ,
\end{align}
where each $\nu_{n,\b}(t,\dd z)$ satisfies condition \eqref{e12}.  Moreover, $M_{0}=0$, $\nu_{0}\ge0
$ and
\begin{align}\label{e21}
   \int_{|z|\ge 1}\ee^{\lambda |z|}\nu_{0}(t,\dd z)<\infty,\qquad t\in[0,T],
\end{align}
for some positive $\lambda$.
%\item[iii)] we have convergence:
%\begin{align}
%\nu
%    &=  \sum_{n=0}^\infty \nu_n , &
%a_\alpha
%    &=  \sum_{n=0}^\infty a_{\alpha,n}, &
%|\alpha|
%    &\leq 2, \label{eq:a.sum}
%\end{align}
%in some sense (pointwise or in norm).
%\item[(ii)] for any $t \in [0,T]$ and $B \in \Bc(\Rb_0^d)$ the functions $\nu_n(t,\cdot,B)$ and $a_{\alpha,n}(t,\cdot)$ are polynomials with
%\begin{align}
%\nu_0(t,\cdot,B)
%    &=\nu_0(t,B) &
%    &\text{and}&
%a_{\alpha,0}(t,\cdot)
%    &=a_{\alpha,0}(t),
%\end{align}
%\item[(iii)] for any $x\in\R^{d}$ the functions $a_{\alpha,n}(\cdot,x)$ belong to
%$L^{\infty}([0,T])$.
\end{enumerate}
Then we say that $\( \Ac_n(t) \)_{0\le n \le N}$, defined by
\begin{align}
%\Ac_n(t,x)&=
 \Ac_n(t)f(x) &= \sum_{|\beta|\le M_{n}}x^{\b} \int_{\Rb^d} \nu_{n,\beta}(t,\dd z) \left( \ee^{\<z,\nabla_x \>} - 1 - \<z,\nabla_x \>\right)f(x)
            + \sum_{|\alpha |\leq 2}  a_{\alpha,n}(t,x) D^{\alpha}_x f(x) \\
            &\equiv  \int_{\Rb^d} \nu_n(t,x,\dd z) \left( \ee^{\<z,\nabla_x \>} - 1 - \<z,\nabla_x \>\right)f(x)
            + \sum_{|\alpha |\leq 2}  a_{\alpha,n}(t,x) D^{\alpha}_x f(x),      \label{eq:An}
\end{align}
is an \emph{$N$th order polynomial expansion} of $\Ac(t)$.
\end{definition}
%Assuming that \eqref{eq:a.sum} holds, the operator $\Ac(t)$ in \eqref{operator_A} can be formally written as
%\begin{align}
%\Ac(t)
%    &=  \Ac_0(t) + \Bc(t) , &
%\Bc(t)
 %   &=  \sum_{n=1}^\infty \Ac_n(t),  \label{eq:A.expand}
%\end{align}
%where
%\begin{align}
%\Ac_n(t)
%    &:= \int_{\Rb_0^d} \nu_n(t,x,\dd z) ( \ee^{\<z,\nabla_x \>} - 1 - \<z,\nabla_x \>)
%            + \sum_{|\alpha |\leq 2}  a_{\alpha,n}(t,x) D^{\alpha}_x  . \label{eq:An}
%\end{align}
\noindent %Notice that Definition \ref{def:A} allows for very general polynomial specifications.
%Clearly approximation results should expected when the expansion $\Ac_n(t)$ is somehow related to
%the original operator $\Ac(t)$: that is the case, for instance, of the following Taylor-based
%polynomial expansions.
Definition \ref{def:A} allows for very general polynomial specifications.  The idea is to choose
an expansion $(\Ac_n(t))$ that closely approximates $\Ac(t)$. The precise sense of this
approximation will depend on the application.  Below, we present three polynomial expansions.  The
first two expansion schemes provide an accurate approximation $\Ac(t)$ in a pointwise local sense,
under the assumption of smooth coefficients. The last expansion scheme approximates $\Ac(t)$ in a
global sense and can be applied even in the case of discontinuous coefficients.
\begin{example}\label{example:Taylor}(Taylor polynomial expansion)\\
Assume the coefficients $a_{\alpha}(t,\cdot)\in C^{N}(\mathbb{R}^d)$ and that the compensator
$\nu$ takes the form
 $$\nu(t,x,\dd z)=h(t,x,z)\bar{\nu}(\dd z)$$
where $h(t,\cdot,z)\in C^{N}(\mathbb{R}^d)$ with $h\geq 0$, and $\bar{\nu}$ is a L\'evy measure.
%is
%for any Borel set $B \in \Bc(\Rb^d)$ the function $\nu(t,\cdot,B)$ to be differentiable, up to order $N$.
Then, for any fixed $\bar{x}\in\R^{d}$ and $n\le N$, we define $\nu_n$ and $a_{\alpha,n}$ as the
$n$th order term of the Taylor expansions of $\nu$ and $a_{\alpha}$ respectively in the spatial
variables $x$ around the point $\xbar$.  That is, we set
\begin{align}
\nu_n(t,x,\dd z)
  &=        \sum_{|\b|=n}\frac{D_{x}^{\b}h(t,\bar{x},z)}{\b!}(x-\bar{x})^{\b}\bar{\nu}(\dd z), %&0    &\leq n\leq N,
  \\
  a_{\alpha,n}(t,x) &= \sum_{|\b|=n}\frac{D_{x}^{\b}a_{\alpha}(t,\bar{x})}{\b!}(x-\bar{x})^{\b}, %& 0    &\leq n\leq N,
  &|\a| &\leq 2 ,
\end{align}
where as usual {$\b!=\b_{1}!\cdots\b_{d}!$ and $x^\beta = x_1^{\beta_1} \cdots x_d^{\beta_d}$}.
The expansion proposed in \cite{lorig-pagliarani-pascucci-2} and
\cite{lorig-pagliarani-pascucci-3} is the particular case when $\nu\equiv 0$, whereas the
expansion proposed in \cite{lorig-pagliarani-pascucci-1} and \cite{lorig-pagliarani-pascucci-1.5}
is a particular case when $d=1$.
\end{example}
\begin{example}\label{example:TimeTaylor}(Time-dependent Taylor polynomial expansion)\\
% Assume the coefficients $a_{\alpha}(t,\cdot)\in C^{N}(\mathbb{R}^d)$, and that $\nu(t,x,\dd z)=h(t,x,z)m(\dd z)$ with $h(t,\cdot,z)\in C^{N}(\mathbb{R}^d)$.
Under the assumptions of Example \ref{example:Taylor},
fix a trajectory $\xb:\Rb_+ \to \Rb^d$.
%let the basis point of the
%expansion to depend on time: $\xbar=\xbar(t)$.
We then define $\nu_n(t,x,\dd z)$ and $a_{\alpha,n}(t,x)$ as the $n$th order term of the Taylor expansions of $\nu(t,x,\dd z)$ and $a_{\alpha}(t,x)$ respectively around $\xbar(t)$. More precisely, we set
\begin{align}
 \nu_n(t,x,\dd z)
    &=\sum_{|\b|=n}\frac{D_x^{\b}h(t,\bar{x}(t), z)}{\b!}(x-\bar{x}(t))^{\b}\bar{\nu}(\dd z),\\
     a_{\alpha,n}(t,x)
    &=\sum_{|\b|=n}\frac{D_x^{\b}a_{\alpha}(t,\bar{x}(t))}{\b!}(x-\bar{x}(t))^{\b}, &|\alpha| &\leq 2 .
\end{align}
This expansion for the coefficients allows the expansion point $\bar{x}$ of the Taylor
series to evolve in time according to the evolution of the underlying process $X_t$. For instance, one could choose $\bar{x}(t)=\mathbb{E}[X_t]$. In \cite{lorig-pagliarani-pascucci-2} this choice results in a highly accurate approximation for option prices and implied volatility in the \cite{heston} model.
\end{example}
%\ora{******** END CHANGES HERE ************}

\begin{example}\label{example:Hilbert}(Hermite polynomial expansion)\\
Hermite expansions can be useful when the diffusion coefficients are discontinuous.  A remarkable
example in financial mathematics is given by the Dupire's local volatility formula for models with
jumps (see \cite{frizyor2013}). In some cases, e.g., the well-known Variance-Gamma model, the
fundamental solution (i.e., the transition density of the underlying stochastic model) has
singularities.  In such cases, it is natural to approximate it in some $L^{p}$ norm rather than in
the pointwise sense. For the Hermite expansion centered at $\xb$, one sets
\begin{align}
\nu_{n}(t,x,\dd z)
    &=  \sum_{|\b|=n}
                \< \Hv_\b(\cdot-\bar{x}) , \nu(t,\cdot,\dd z) \>_{\Gam} \Hv_\b(x-\bar{x}),\\
a_{\alpha,n}(t,x)
    &=  \sum_{|\b|=n}
                \< \Hv_\b(\cdot-\bar{x}) , a_{\alpha}(t,\cdot) \>_{\Gam} \Hv_\b(x-\bar{x}), & |\a|
    &\leq 2 ,
\end{align}
where the inner product $\<\cdot,\cdot\>_\Gam$ is an integral over $\Rb^d$ with a Gaussian
weighting centered at $\xb$ and the functions $\Hv_\beta(x) = H_{\beta_1}(x_1) \cdots
H_{\beta_d}(x_d)$ where $H_n$ is the $n$-th one-dimensional Hermite polynomial (properly
normalized so that $\< \Hv_\alpha, \Hv_\beta \>_\Gam = \delta_{\alpha,\beta}$ with
$\delta_{\alpha,\beta}$ being the Kronecker's delta function).
\end{example}

%%%%%%%%%%%%%%%%%%%%%%%%%%%%%%%%%%%%%%%
%
%       SECTION: Dyson Series
%
%%%%%%%%%%%%%%%%%%%%%%%%%%%%%%%%%%%%%%%

\section{Formal solution via Dyson series}
\label{sec:dyson}
%Throughout this Section, we will explicitly indicate $t$-dependence in all operators.  We will generally hide $x$-dependence, except where it is needed for clarity.
%\ora{Given a polynomial expansion $(\Ac_n(t))_{n\geq 0}$ satisfying Definition \ref{def:A}, the operator $\Ac(t)$ in \eqref{operator_A} can be formally written as}
In this section we present a heuristic argument to pass from an expansion of the operator $\Ac(t)$
in \eqref{eq:A} to an expansion for $u$, the solution of problem \eqref{eq:u.pide}. The following
argument is not intended to be rigorous.  Rather, the computations that follow provide motivation for the price expansion given in Definition \ref{def:ub.N}.
Throughout this section, we will generally omit $x$-dependence, except where it is needed for clarity.  To begin, we presume that the operator $\Ac(t)$ can be formally written as
\begin{align}
\Ac(t)
    &=  \Ac_0(t) + \Bc(t) , &
\Bc(t)
    &=  \sum_{n=1}^\infty \Ac_n(t). \label{eq:A.expand} %&
%\Ac_n(t)
%   &:= \sum_{|\alpha |\leq 2}  a_{\alpha,n}(t,x) D_x^{\alpha} .
\end{align}
%Let us consider a polynomial expansion $\left(a_{\alpha,n}\right)$ of $\Ac(t)$.
We insert expansion \eqref{eq:A.expand} for $\Ac(t)$ into Cauchy problem \eqref{eq:u.pide} and find
\begin{align}
( \d_t + \Ac_0(t) ) u(t)
    &=  - \Bc(t) u(t) , &
u(T)
    &=  \varphi .
\end{align}
Note that, by construction, $\Ac_0(t)$ is the generator of an additive process. % (time inhomogeneous L\'evy process)} \cite[Section 14]{contbook}.
Therefore, by Duhamel's principle, we have
\begin{align}
u(t)
    &=  \Pc_0(t,T) \varphi + \int_t^T \dd t_1 \, \Pc_0(t,t_1) \Bc(t_1) u(t_1) , \label{eq:u.duhemel}
\end{align}
where $\Pc_0(t,T)$ is the semigroup of operators generated by $\Ac_0(t)$. %, is given by
%\begin{align}
%\Pc_0(t,T)
%    &=  \exp \int_t^T \dd s \, \Ac_0(s) . \label{eq:P0}
%\end{align}
Inserting expression \eqref{eq:u.duhemel} for $u$ into the right-hand side of \eqref{eq:u.duhemel} and iterating we obtain
\begin{align}
u(t)
    &=  \Pc_0(t,T) \varphi
            + \int_{t}^T \dd t_1 \, \Pc_0(t,t_1) \Bc(t_1) \Pc_0(t_1,T) \varphi \\ & \qquad
            + \int_{t}^T \dd t_1 \int_{t_1}^T \dd t_2 \, \Pc_0(t,t_1) \Bc(t_1)
                \Pc_0(t_1,t_2) \Bc(t_2) u(t_2) \\
    &=  \cdots \\
    &=  \Pc_0(t,T) \varphi + \sum_{k=1}^\infty
            \int_{t}^T \dd t_1 \int_{t_1}^T \dd t_2 \cdots \int_{t_{k-1}}^T \dd t_k
            \\ & \qquad
            \Pc_0(t,t_1) \Bc(t_1)
            \Pc_0(t_1,t_2) \Bc(t_2) \cdots
            \Pc_0(t_{k-1},t_k) \Bc(t_k)
            \Pc_0(t_k,T) \phi
            \label{eq:dyson} \\
    &=  \Pc_0(t,T) \varphi + \sum_{n=1}^\infty \sum_{k=1}^n
            \int_{t}^T \dd t_1 \int_{t_1}^T \dd t_2 \cdots \int_{t_{k-1}}^T \dd t_k
            \\ & \qquad \sum_{i \in I_{n,k}}
            \Pc_0(t,t_1) \Ac_{i_1}(t_1)
            \Pc_0(t_1,t_2) \Ac_{i_2}(t_2) \cdots
            \Pc_0(t_{k-1},t_k) \Ac_{i_k}(t_k)
            \Pc_0(t_k,T)\phi ,
            \label{eq:u.dyson} \\
I_{n,k}
    &= \{ i = (i_1, i_2, \cdots , i_k ) \in \mathbb{N}^k \mid i_1 + i_2 + \cdots + i_k = n \}.
            \label{eq:Ink}
\end{align}
The second-to-last equality \eqref{eq:dyson} is known as the {\it Dyson series expansion} of $u$
(see, for instance, Section 5.7 of \cite{sakurai1994modern} or Chapter IX.2.6 of \cite{kato}). To
obtain \eqref{eq:u.dyson} from \eqref{eq:dyson} we have used \eqref{eq:A.expand} to replace $\Bc(t)$ by the infinite sum $\sum_{n=1}^\infty \Ac_n(t)$, and we have partitioned on the sum of the subscripts of the $(\Ac_{i_k})$.
%defined as the infinite sum in \eqref{eq:A.expand}.
Expansion \eqref{eq:u.dyson} motivates the following definition.
\begin{definition}
\label{def:ub.N} For a fixed $N$th order polynomial expansion $(\Ac_n(t))_{0\le n\le N}$ satisfying Definition \ref{def:A}, we define $\ub_N$, the \emph{$N$th order price approximation} of $u$,
%the solution of \eqref{eq:u.pide},
as
\begin{align}
\ub_N
    &:=  \sum_{n=0}^N u_n ,  \label{eq:ubar.N} %\label{eq:u0.def}
\end{align}
where
\begin{align}
 u_0(t)
    &:= \Pc_0(t,T) \varphi ,\\
 u_n(t)
    &:= \sum_{k=1}^n
            \int_{t}^T \dd t_1 \int_{t_1}^T \dd t_2 \cdots \int_{t_{k-1}}^T \dd t_k
            \\ & \qquad \sum_{i \in I_{n,k}}
            \Pc_0(t,t_1) \Ac_{i_1}(t_1)
            \Pc_0(t_1,t_2) \Ac_{i_2}(t_2) \cdots
            \Pc_0(t_{k-1},t_k) \Ac_{i_k}(t_k)
            \Pc_0(t_k,T) \phi,\qquad n\ge1. \label{eq:un.def}
\end{align}
Here, $\Pc_0(t,T)$ is the semigroup generated by $\Ac_0(t)$ and $I_{n,k}$ is as given in \eqref{eq:Ink}.
\end{definition}
%\begin{remark}
%Observe that the $N$th order price approximation $\ub_N$ requires only an $N$th order polynomial expansion of $\Ac(t)$.
%\end{remark}
\noindent In Sections \ref{sec:u.0} and \ref{sec:u.n} we will provide explicit expressions
for $u_0$ and $\left(u_n\right)_{n \geq 1}$ respectively.
%\begin{remark}
%We introduce $\hat{P}$ and $\ee_\xi$, the \emph{characteristic function} and \emph{oscillating exponential}, respectively
%\begin{align}
%\hat{P}(t,x,T,\xi)
    %&:=    \Eb[ \ee^{\int_t^T a_{(0,\cdots,0),0}(s,X_s)\dd s} \ee^{\ii \xi X_T} | X_t = x] , &
%\ee_\xi(x)
    %&= \ee^{\ii \xi x} .
%\end{align}
%From \eqref{eq:v} we observe that $\hat{P}(t,x,T,\xi)$ is obtained as the special case $\phi=\ee_\xi$.  We denote by $\hat{P}_n(t,x,T,\xi)$ the corresponding $n$th order approximation of $\hat{P}(t,x,T,\xi)$.
%\end{remark}

\subsection{Expression for $u_0$}
\label{sec:u.0} In what follows, it will be helpful to recall the definition of the Fourier and
inverse Fourier transforms.  For any function $\phi$ in the Schwartz class, we define
\begin{align}
\text{Fourier transform:}&& \Fc[\phi](\xi) = \phih(\xi)
    &=  \int_{\Rb^d} \dd x \, \phi(x)  \ee^{ \ii \< \xi , x\> } , \label{eq:ft} \\
\text{Inverse transform:}&& \Fc^{-1}[\hat{\phi}](x) = \phi(x)
    &=  \frac{1}{(2 \pi )^d}\int_{\Rb^d} \dd \xi \, \phih(\xi) \ee^{- \ii \< \xi , x\> } .
            \label{eq:ift}
\end{align}
%As noted above, the operator $\Ac_0(t)$ is the infinitesimal generator of an additive process. We
%denote by $X^{(0)}$ a process with generator $\Ac_0(t)$.
Recall that by construction $M_{0}=0$ (cf. Definition \ref{def:A}) and therefore the operator
$\Ac_0(t)$ has time-dependent
coefficients which are independent of $x$. Then the action of the semigroup of operators $\Pc_0(t,T)$ of $\Ac_0(t)$ % on a function $\phi$ belonging to the Schwartz class
is well-known:
%\begin{align}
%u_0(t)
%    :=  \Pc_0(t,T) \phi
%    &=\frac{1}{(2 \pi )^d}  \int_{\Rb^d}
%            \ee^{\ii \< \xi, x \> + \Phi_0(t,T,\xi)} \phih(-\xi)\, \dd \xi , \label{eq:u0}
%\end{align}
%where $\phih$ is the Fourier transform of $\phi$ and $\Phi_0$ is given by
%\begin{align}
%\Phi_0(t,t_k,\xi)
%    &= \Psi(t,t_k,\xi)
%            + \sum_{|\alpha |\leq 2} (\ii \xi)^\alpha \int_{t}^{t_k} \dd t \, a_{\alpha,0}(t),
%            \label{eq:Phi0}
%\end{align}
%with
%\begin{equation}\label{eq:Psi}
%\Psi(t,t_k,\xi)=  \int_{t}^{t_k} \int_{\Rb_0^d}
%             \( \ee^{\ii \< \xi, z \>} - 1 - \ii \< \xi, z \> \)\nu_0(s,\dd z)\dd s.
%\end{equation}
\begin{align}
  u_0(t):=  \Pc_0(t,T) \phi
    &=\frac{1}{(2 \pi )^d}  \int_{\Rb^d}
            \hat{P}_0(t,x,T,\xi)\phih(-\xi)\, \dd \xi \, \label{eq:u0}
\end{align}
%\ora{
where %$\phih$ is the Fourier transform of $\phi$ and
\begin{align}\label{eq:P0}
  %\hat{P}_0(t,T,\xi):=
  \hat{P}_0(t,x,T,\xi)
        &:= \ee^{\ii \< \xi, x \> + \Phi_0(t,T,\xi)}
\end{align}
with
%\begin{align}
%\Phi_0(t,T,\xi)
%    &=\sum_{|\alpha |\leq 2} (\ii \xi)^\alpha \int_{t}^{T} \dd s \, a_{\alpha,0}(s) + \int_{t}^{T} \int_{\Rb_0^d}
%             \nu_0(s,\dd z)\dd s\, \( \ee^{\ii \< \xi, z \>} - 1 - \ii \< \xi, z \> \).
%            \label{eq:Phi0}
%\end{align}
\begin{align}
\Phi_0(t,T,\xi)
    &= \sum_{|\alpha |\leq 2} (\ii \xi)^\alpha \int_{t}^{T} \dd s \, a_{\alpha,0}(s) + \Psi_{0}(t,T,\xi),
            \label{eq:Phi0}
\end{align}
and
\begin{align}\label{eq:Psi0}
\Psi_{0}(t,T,\xi)=  \int_{t}^{T} \int_{\Rb^d}
             \( \ee^{\ii \< \xi, z \>} - 1 - \ii \< \xi, z \> \)\nu_0(s,\dd z)\dd s.
\end{align}
%In probabilistic terminology, let $X^{(0)}$ be the process with generator $\Ac_0(t)-a_{0,0}(s)$:
%then $\hat{P}_0(t,x,T,\xi)$ is the characteristic function of the process $X^{(0)}$ ``with
%default'', that is
%  $$\hat{P}_0(t,x,T,\xi)=\Eb\[ \ee^{\int_t^T a_{0,0}(s)\dd s} \ee^{ \ii \xi X^{(0)}_{T}}\mid X^{(0)}_{t}=x\].$$
\begin{remark}\label{r5}
%Note that the characteristic function $\Eb[ \ee^{\ii \xi X_T} | X_t = x]$.
%Note that $\hat{P}_0(t,x,T,\xi)$ represents the \emph{$0$th order approximation of the characteristic function of $X^{(0)}$}.  More generally,
We introduce $\hat{P}$ and $\ee_\xi$, the {\it characteristic function} and \emph{oscillating
exponential}, respectively
\begin{align}\label{osci}
\hat{P}(t,x,T,\xi)
    &:= \Eb \left[ \ee^{\int_t^T a_{0,0}(s,X_s)\dd s} \ee^{\ii \<\xi, X_T\>} | X_t = x\right] , &
\ee_\xi(x)
    &=  \ee^{ \ii \< \xi , x\>} ,
\end{align}
where $a_{0,0}$ is short-hand for $a_{(0,0,\cdots,0),0}$.  From \eqref{eq:v} we observe that $\hat{P}(t,x,T,\xi)$ is obtained as the special case $\phi=\ee_\xi$.  We note that $\hat{P}_0(t,x,T,\xi)$ in \eqref{eq:P0} represents the $0$th order approximation of $\hat{P}(t,x,T,\xi)$.  More generally, we denote by $\hat{P}_n(t,x,T,\xi)$ the $n$th order approximation of $\hat{P}(t,x,T,\xi)$, obtained by setting $\phi=\ee_\xi$ in \eqref{eq:un.def}.
%precisely, this is defined by
%setting $\phi=\ee_\xi$ in \eqref{eq:v} where $\ee_\xi$ is the oscillating exponential
%\begin{align}\label{osci}
%\ee_\xi(y)
    %&:= \ee^{ \ii \< \xi , y\> } .
%\end{align}
\end{remark}
%}
%
%{\bf\blue [Andrea: formula \eqref{eq:u0} is the well-known Pacerval isometry
%  $$\int_{\Rb^d}{P}_0(t,x,T,y)\phi(y)\, \dd y =\frac{1}{(2 \pi )^d}  \int_{\Rb^d}
%            \hat{P}_0(t,x,T,\xi)\overline{\phih(\xi)}\, \dd \xi$$
%Notice that $\phih(-\xi)=\overline{\phih(\xi)}$: the negative sign in $\phih(-\xi)$ must be there,
%for any definition of the Fourier transform.]}

%and
%\begin{equation}\label{eq:Psi0}
%\Psi_{0}(t,T,\xi)=  .
%\end{equation}

\subsection{Expression for $u_n$}
\label{sec:u.n} Remarkably, as the following proposition shows, every $u_n(t)$ can be expressed as
a pseudo-differential operator $\Lc_n(t,T)$ acting on $u_0(t)$.
\begin{proposition}\label{thm:dyson}
Assume that $\phi$ belongs to the Schwartz class, and that $\Phi_{0}$ in \eqref{eq:Phi0} is a smooth function of the variable $\xi$. Then the function $u_n$ defined in \eqref{eq:un.def} is given explicitly by
\begin{align}
u_n(t)
    &=  \Lc_n(t,T) u_0(t) , \label{eq:un} &
\end{align}
where $u_0$ is given by \eqref{eq:u0} and
\begin{align}
\Lc_n(t,T)
    &=  \sum_{k=1}^n
            \int_{t}^T \dd t_1 \int_{t_1}^T \dd t_2 \cdots \int_{t_{k-1}}^T \dd t_k
            \sum_{i \in I_{n,k}}
            \Gc_{i_1}(t,t_1)
            \Gc_{i_2}(t,t_2) \cdots
            \Gc_{i_k}(t,t_k) , \label{eq:Ln}
\end{align}
with $I_{n,k}$ as defined in \eqref{eq:Ink} and
\begin{align}
\Gc_j(t,t_k)
    &:= \Ac_j(t_k,\Mc(t,t_k)) \\
        &=  \int_{\Rb^d}\nu_j(t_k,\Mc(t,t_k),\dd z)
            \left(\ee^{ \<z, \nabla_x \>} - 1 - \<z, \nabla_x \>\right)
            + \sum_{|\alpha |\leq 2}  a_{\alpha,j}(t_k,\Mc(t,t_k)) D_x^\alpha ,
    \label{eq:G.def} \\
%\Mc_i(t,t_k)
    %&:= x_i + \int_{\Rb_0^d} \int_{t}^{t_k} \dd s \,
            %\nu_0(s,\dd z) z_i \( \ee^{\ii \< z,\nabla_x \>} - 1 \)
            %+ \sum_{|\alpha |\leq 2} \int_{t}^{t_k} \dd s \, a_{\alpha,0}(s) \cdot
            %\ii^{|\alpha|} (-\alpha_i) \nabla_x^{\alpha-\eps_i} ,
            %\label{eq:M.def} \\
%\eps_i
    %&:= (0, 0, \cdots, \underbrace{1}_{\textrm{component $i$}}, \cdots, 0, 0) .
%\end{align}
\Mc(t,t_k)
    &:= x + \int_{\Rb^d} \int_{t}^{t_k} z \( \ee^{\< z,\nabla_x \>} - 1 \)
            \nu_0(s,\dd z)\dd s
            + \int_{t}^{t_k}  m(s)\dd s + \int_{t}^{t_k}  C(s) \nabla_x \dd s  ,
                        \label{eq:M.def} \\
m(s)
        &=  \begin{pmatrix}
                a_{(1,0,\dots,0),0}(s) & a_{(0,1,\dots,0),0}(s) & \ldots &  a_{(0,0,\dots,1),0}(s)
                \end{pmatrix} ,\label{eq:m} \\
C(s)
        &= \begin{pmatrix}
                2a_{(2,0,\dots,0),0}(s) & a_{(1,1,\dots,0),0}(s) & \ldots &  a_{(0,0,\dots,1),0}(s) \\
                a_{(1,1,\dots,0),0}(s) & 2a_{(0,2,\dots,0),0}(s) & \ldots &  a_{(0,1,\dots,1),0}(s) \\
                \vdots & \vdots & \ddots & \vdots \\
                a_{(1,0,\dots,1),0}(s) & a_{(0,1,\dots,1),0}(s) & \ldots &  2 a_{(0,0,\dots,2),0}(s) \\
                \end{pmatrix}\label{eq:C} .
\end{align}
%{\blue Stefano: This expression for $\Mc$ is not clear. $\Mc$ is a vector, but the term
%\begin{equation}\label{eq:intwrong}
%\int_{\Rb_0^d} \int_{t}^{t_k}z \( \ee^{\ii \< z,\nabla_x \>} - 1 \)\nu_0(s,\dd z)\dd s
%\end{equation}
%is a scalar. It should be replaced with the vector
%$$
%\left(  \int_{\Rb_0^d} \int_{t}^{t_k} z_1 \( \ee^{\ii \< z,\nabla_x \>} - 1 \)\nu_0(s,\dd z)\dd s  ,\cdots, \int_{\Rb_0^d} \int_{t}^{t_k}  z_d \( \ee^{\ii \< z,\nabla_x \>} - 1 \) \nu_0(s,\dd z)\dd s   \right)
% $$
% I understand that the notation is very heavy, but \eqref{eq:intwrong} is not correct..
% }
Moreover, the components of $\Mc(t,t_k)$ commute.  Therefore the operators $(\Gc_j(t,t_k))$, which are polynomials in $\Mc(t,t_k)$ by construction, are well defined.
\end{proposition}
\begin{proof}
%{\blue First we remark that $\Gc_j(t,t_k)$ in \eqref{eq:G.def} is well defined as polynomials of
%$\Mc(t,t_k)$ because the components $\Mc_{i}(t,t_k)$, $i=1,\dots,d$, of the operator $\Mc(t,t_k)$
%commute (see \eqref{e50} below).} Next
The proof consists in showing that the operator
$\Gc_j(t,t_k)$ in \eqref{eq:G.def} satisfies
\begin{align}
\Pc_0(t,t_k) \Ac_{j}(t_k)
    &=  \Gc_j(t,t_k) \Pc_0(t,t_k) . \label{eq:PA=GP}
\end{align}
Assuming \eqref{eq:PA=GP} holds, we can use the fact that $\Pc_0(t_k,t_{k+1})$ is a semigroup
\begin{align}
\Pc_0(t,T)
    &=  \Pc_0(t,t_1) \Pc_0(t_1,t_2) \cdots \Pc_0(t_{k-1},t_k) \Pc_0(t_k,T) , &
t
    &\leq t_1 \leq \ldots \leq t_k \leq T ,
\end{align}
and we can re-write \eqref{eq:un.def} as
\begin{align}
u_n(t)
    &=  \sum_{k=1}^n
            \int_{t}^T \dd t_1 \int_{t_1}^T \dd t_2 \cdots \int_{t_{k-1}}^T \dd t_k
            \sum_{i \in I_{n,k}}
            \Gc_{i_1}(t,t_1)
            \Gc_{i_2}(t,t_2) \cdots
            \Gc_{i_k}(t,t_k)
            \Pc_0(t,T) \phi , \label{eq:u.dyson2}
\end{align}
from which \eqref{eq:un}-\eqref{eq:Ln} follows directly.  Thus, we only need to show that
$\Gc_j(t,t_k)$ satisfies \eqref{eq:PA=GP}. % Without loss of generality, by \eqref{eq:ift}, we can
It is sufficient to investigate how the operator $\Pc_0(t,t_k) \Ac_{j}(t_k)$ acts on the
oscillating exponential in \eqref{osci}. First, we note that
\begin{align}
\Pc_0(t,t_k) \ee_\xi(x)
    &=  \ee^{\Phi_0(t,t_k,\xi)} \ee_\xi(x) ,          \label{eq:P.exp}
\end{align}
where $\Phi_0(t,t_k,\xi)$, as given in \eqref{eq:Phi0}, is a smooth function by condition
\eqref{e21}. Next, we observe that the operator $\Mc(t,t_k)$ in \eqref{eq:M.def} can be written
%\begin{align}
%\Mc_i(t,t_k)
    %&=  M_i(t,t_k, -\ii \nabla_x ) , &
%M_i(t,t_k, \xi )
    %&=  - \ii \d_{\xi_i} \( \Phi_0(t,t_k,\xi) + \ii \< \xi, x \> \) . \label{eq:M}
%\end{align}
\begin{align}
\Mc(t,t_k)
    &=  M(t,t_k, -\ii \nabla_x ) , &
M(t,t_k, \xi )
    &=  - \ii \nabla_\xi \( \Phi_0(t,t_k,\xi) + \ii \< \xi, x \> \) . \label{eq:M}
\end{align}
Denote by $\Mc_j$ and $M_j$ the $j$th component of $\Mc$ and $M$ respectively.  Then, using
\eqref{eq:M} %we observe that , for any $n\le N$, we have
%%{\blue \\ Stefano: We don't know how many times $\Phi_0(t,t_k,\xi)$ is differentiable. It seems to me that it depends on the tails of $\nu_0(\dd z)$. In particular, it seems that $\partial^n_{\xi}\Phi_0(t,t_k,\xi)$ is well defined if $\int_{\mathbb{R}}z^n \dd z <\infty$. Shall we add the condition $\int_{\mathbb{R}}(e^z +e^{-z}) \dd z <\infty$ to ensure that $\partial^n_{\xi}\Phi_0(t,t_k,\xi)$ is well defined for any $n$?}
%\begin{align}
%(-\ii \d_{\xi_j})^n \ee^{\Phi_0(t,t_k,\xi)} \ee_\xi(x)
%    &=  (-\ii \d_{\xi_j})^{n-1} M_j(t,t_k,\xi )
%            \ee^{\Phi_0(t,t_k,\xi)} \ee_\xi(x) \\
%    &=  \Mc_j(t,t_k) (-\ii \d_{\xi_j})^{n-1}
%            \ee^{\Phi_0(t,t_k,\xi)} \ee_\xi(x) \\
%    &=  \cdots \\
%    &=   \left(\Mc_j(t,t_k)\right)^n \ee^{\Phi_0(t,t_k,\xi)} \ee_\xi(x) .\label{e50}
%\end{align}
we have
%{\blue \\ Stefano: We don't know how many times $\Phi_0(t,t_k,\xi)$ is differentiable. It seems to me that it depends on the tails of $\nu_0(\dd z)$. In particular, it seems that $\partial^n_{\xi}\Phi_0(t,t_k,\xi)$ is well defined if $\int_{\mathbb{R}}z^n \dd z <\infty$. Shall we add the condition $\int_{\mathbb{R}}(e^z +e^{-z}) \dd z <\infty$ to ensure that $\partial^n_{\xi}\Phi_0(t,t_k,\xi)$ is well defined for any $n$?}
\begin{align}
(-\ii \d_{\xi_i})(-\ii \d_{\xi_j}) \ee^{\Phi_0(t,t_k,\xi)} \ee_\xi(x)
    &=  (-\ii \d_{\xi_i}) M_j(t,t_k,\xi )
            \ee^{\Phi_0(t,t_k,\xi)} \ee_\xi(x) \\
    &=  \Mc_j(t,t_k) (-\ii \d_{\xi_i})
            \ee^{\Phi_0(t,t_k,\xi)} \ee_\xi(x) \\
    &=  \Mc_j(t,t_k) M_i(t,t_k,\xi )
            \ee^{\Phi_0(t,t_k,\xi)} \ee_\xi(x) \\
    &=   \Mc_j(t,t_k)\Mc_i(t,t_k) \ee^{\Phi_0(t,t_k,\xi)} \ee_\xi(x) .\label{e50}
\end{align}
More generally %, noting that $\d_{\xi_i}$ and $\d_{\xi_j}$ commute,
for any multi-index $\beta$ we have
\begin{align}
(-\ii \nabla_\xi)^\beta \ee^{\Phi_0(t,t_k,\xi)} \ee_\xi(x)
    &=  (\Mc(t,t_k))^\beta \ee^{\Phi_0(t,t_k,\xi)} \ee_\xi(x) . \label{eq:a.M}
\end{align}
From \eqref{e50} we deduce that operators $\Mc_i$ and $\Mc_j$ commute when applied to
$\ee^{\Phi_0(t,t_k,\xi)} \ee_\xi(x)$, because so do $\d_{\xi_i}$ and $\d_{\xi_j}$.
Consequently, $\Mc_i$ and $\Mc_j$ also commute when applied to $\ee_\xi(x)$ or any
function that admits a representation as a Fourier transform.  To see this observe that
\begin{align}
\Mc_j(t,t_k)\Mc_i(t,t_k) \ee^{\Phi_0(t,t_k,\xi)} \ee_\xi(x)
    &=  \Mc_i(t,t_k)\Mc_j(t,t_k) \ee^{\Phi_0(t,t_k,\xi)} \ee_\xi(x) .
\end{align}
Therefore, since $\Mc_j(t,t_k)$ acts on $x$ and not $\xi$ we have
\begin{align}
\Mc_j(t,t_k)\Mc_i(t,t_k)\ee_\xi(x)
    &=  \Mc_i(t,t_k)\Mc_j(t,t_k) \ee_\xi(x) .
\end{align}
%\ora{This is not clear to me.  We apply $(-\ii \d_{\xi_i})(-\ii \d_{\xi_j})$ to
%$\ee^{\Phi_0(t,t_k,\xi)} \ee_\xi(x)$, not $\ee_\xi(x)$. {\bf\blue Andrea: \eqref{e50} proves that
  %$$\Mc_j(t,t_k)\Mc_i(t,t_k) \ee^{\Phi_0(t,t_k,\xi)} \ee_\xi(x)=\Mc_i(t,t_k)\Mc_j(t,t_k) \ee^{\Phi_0(t,t_k,\xi)} \ee_\xi(x)$$
%right? But
  %$$\Mc_i(t,t_k) \ee^{\Phi_0(t,t_k,\xi)} \ee_\xi(x)=\ee^{\Phi_0(t,t_k,\xi)} \left(\Mc_i(t,t_k)  \ee_\xi(x)\right)$$
%because $\Mc_i(t,t_k)$ acts on x and therefore $\ee^{\Phi_0(t,t_k,\xi)}$ plays as a positive
%constant. Then we also have
  %$$\Mc_j(t,t_k)\Mc_i(t,t_k)\ee_\xi(x)=\Mc_i(t,t_k)\Mc_j(t,t_k) \ee_\xi(x).$$
%Clearly, this can also be verified by a direct computation.} So, why should the result hold for
%any function that has a Fourier transform? {\blue Andrea: I'd say ``any function that
%\underline{is} a Fourier transform''.} Moreover, none of the steps below require that the
%components of $\Mc$ commute on all functions with Fourier representations. {\bf\blue The
%definition of $\Lc_{n}$ requires that $\Mc$ commute in general, not only for the exponential
%oscillating, because $\Lc_{n}$ is \underline{not} a polynomial of $\Mc$.} The operators $\Mc$ are
%only applied to $\ee^{\Phi_0(t,t_k,\xi)} \ee_\xi(x)$. The essential feature of the proof is that
%the commutation relation $\Pc \Ac = \Gc \Pc$ holds for any function with a Fourier transform. And
%this is what the steps below show.}
Finally, we compute
\begin{align}
\Pc_0(t,t_k) \Ac_{j}(t_k) \ee_\xi(x)
    &=  \Pc_0(t,t_k) \int_{\Rb^d} \nu_j(t_k,x,\dd z)
            (\ee^{\<z,\nabla_x \>} - 1 - \<z,\nabla_x \>) \ee_\xi(x)
            \\ & \qquad
            + \sum_{|\alpha |\leq 2}  \Pc_0(t,t_k)
            a_{\alpha,j}(t_k,x) D^{\alpha}_{x} \ee_\xi(x) &
    &\text{(by \eqref{eq:An})} \\
    &=  \Pc_0(t,t_k) \int_{\Rb^d} (\ee^{\ii \<z, \xi \>} - 1 - \ii \<z, \xi \>)
            \nu_j(t_k,x,\dd z) \ee_\xi(x)
            \\ & \qquad
            + \sum_{|\alpha |\leq 2}  (\ii \xi)^{\alpha} \Pc_0(t,t_k)
            a_{\alpha,j}(t_k,x) \ee_\xi(x) \\
    &=  \int_{\Rb^d} (\ee^{\ii \<z, \xi \>} - 1 - \ii \<z, \xi \>)
            \nu_j(t_k,-\ii \nabla_\xi,\dd z) \Pc_0(t,t_k) \ee_\xi(x)
            \\ & \qquad
            + \sum_{|\alpha |\leq 2}  (\ii \xi)^{\alpha} a_{\alpha,j}(t_k,-\ii \nabla_\xi)
            \Pc_0(t,t_k) \ee_\xi(x) \\
    &=  \int_{\Rb^d} (\ee^{\ii \<z, \xi \>} - 1 - \ii \<z, \xi \>)
            \nu_j(t_k,-\ii \nabla_\xi,\dd z) \ee^{\Phi_0(t,t_k,\xi)} \ee_\xi(x)
            \\ & \qquad
            + \sum_{|\alpha |\leq 2}  (\ii \xi)^{\alpha} a_{\alpha,j}(t_k,-\ii \nabla_\xi)
            \ee^{\Phi_0(t,t_k,\xi)} \ee_\xi(x) &
    &\text{(by \eqref{eq:P.exp})} \\
    &=  \int_{\Rb^d} (\ee^{\ii \<z, \xi \>} - 1 - \ii \<z, \xi \>)
            \nu_j(t_k,\Mc(t,t_k),\dd z) \ee^{\Phi_0(t,t_k,\xi)} \ee_\xi(x)
            \\ & \qquad
            + \sum_{|\alpha |\leq 2}  (\ii \xi)^{\alpha} a_{\alpha,j}(t_k,\Mc(t,t_k))
            \ee^{\Phi_0(t,t_k,\xi)} \ee_\xi(x) &
    &\text{(by \eqref{eq:a.M})} \\
    &=  \int_{\Rb^d}\nu_j(t_k,\Mc(t,t_k),\dd z)
            (\ee^{ \<z, \nabla_x \>} - 1 - \<z, \nabla_x \>)
            \ee^{\Phi_0(t,t_k,\xi)} \ee_\xi(x)
            \\ & \qquad
            + \sum_{|\alpha |\leq 2}  a_{\alpha,j}(t_k,\Mc(t,t_k)) D_x^\alpha
            \ee^{\Phi_0(t,t_k,\xi)} \ee_\xi(x)    \\
    &=  \int_{\Rb^d}\nu_j(t_k,\Mc(t,t_k),\dd z)
            (\ee^{ \<z, \nabla_x \>} - 1 - \<z, \nabla_x \>)
            \Pc_0(t,t_k) \ee_\xi(x)
            \\ & \qquad
            + \sum_{|\alpha |\leq 2}  a_{\alpha,j}(t_k,\Mc(t,t_k)) D_x^\alpha
            \Pc_0(t,t_k) \ee_\xi(x)   &
    &\text{(by \eqref{eq:P.exp})} \\
    &=  \Gc_j(t,t_k) \Pc_0(t,t_k) \ee_\xi(x) , &
    &\text{(by \eqref{eq:G.def})}
\end{align}
which concludes the proof.
\end{proof}
%\begin{definition}
%We define our \emph{$N$-th order price approximation} $\ub_N$ as
%\begin{align}
%\ub_N
    %&=  \sum_{n=0}^N u_n , \label{eq:ubar.N}
%\end{align}
%with $u_0$ and $u_n$ given by \eqref{eq:u0} and \eqref{eq:un} respectively.
%\end{definition}

\begin{remark}
 Error bounds for the Taylor approximation $\ub_N$ in the scalar case $d=1$ can be found in
 \cite{lorig-pagliarani-pascucci-1,lorig-pagliarani-pascucci-1.5}.
\end{remark}

\subsection{Fourier representation for $u_n$}
%\ora{I think we should keep this as a subsection; it helps to tell the reader that we are moving on to a new computation.  A few other notes\\
%1.  I tried to follow the computation below and it is confusing to me, so it will be even more confusing for the reader.  Please try to re-write this Section more clearly.  I would suggest writing it in Theorem-Lemma-Proof style as I have written Theorem \ref{thm:dyson} and its proof.  I will be happy to check it again when you are done.\\
%2.  Please use $\ee$ instead of $e$ to be consistent with the rest of the paper.\\
%3.  $\widehat{\Mc}_{i}^{\xi}(t,t_k)$ is introduced without ever mentioning what it is.  Also, we should not call it the \emph{symbol} because it is an operator -- not a symbol.
%}
%\noindent
%Note from \eqref{eq:Ln} that $\Lc_n(t,T)$ is an integro-differential operator.  Thus, computing $u_n$ using \eqref{eq:un} requires both integration and differentiation.  This can be done explicitly since,
Using \eqref{eq:u0}, \eqref{eq:P0} and \eqref{eq:un} we have
\begin{align}
u_n(t,x) = \Lc_n(t,T) u_0(t,x)
    &=  \frac{1}{(2\pi)^d} \int_{\Rb^d} \ee^{\Phi_0(t,T,\xi)}\(\Lc_n(t,T)\ee^{\ii \<\xi,x\>} \)
            \hat{\phi}(-\xi) \dd \xi.
\end{align}
The term in parenthesis $\Lc_n(t,T) \ee^{\ii \<\xi,x\>}$ can be computed explicitly.  However,
$\Lc_n(t,T)$ is, in general, an \emph{integro-differential} operator (when $X$ is a diffusion
$\Lc_n(t,T)$ is simply a differential operator).  Thus, for models with jumps, computing
$\Lc_n(t,T)\ee^{\ii \<\xi,x\>}$ is a challenge.  Remarkably, we will show that there exists a
first order \emph{differential} operator $\hat{\Lc}^{\xi}_n(t,T)$ such that
%\begin{align}
%\Lc_n(t,T)\ee^{\ii \<\xi,x\>}
    %&= \hat{\Lc}_n(t,T) \ee^{\ii \<\xi,x\>}
%\end{align}
%where $\hat{\Lc}_n(t,T)$ is a \emph{differential operator}, which acts on $\xi$ rather than $x$.
%For computational purposes, it is sometimes convenient to express prices as an inverse Fourier
%transform. The aim of this section is to derive new pricing formulas whose implementation is more
%direct than those given in Proposition \ref{thm:dyson}. This will be done by switching from the
%state variables $x$ to the variables $\xi$ in the Fourier space. Since the operator $\Lc_{n}(t,T)$
%in \eqref{eq:Ln} acts in the variables $x$, for clarity, we shall denote it also by
%$\Lc_{n}^{x}(t,T)$ below.
%
 %%The following theorem provides and expression for $\uh_n$, the Fourier transform of $u_n$ in the price expansion \eqref{eq:ubar.N}.
%Next we define $\hat{\Lc}^{\xi}_{n}(t,T)$ as the operator acting in the variable $\xi$, such that
\begin{align}\label{e33}
  \Lc^{x}_{n}(t,T)\ee^{\ii \<\xi, x \>}=\hat{\Lc}^{\xi}_{n}(t,T)\ee^{\ii \<\xi, x \>} ,
\end{align}
where, for clarity, we have explicitly indicated using superscripts that $\Lc^{x}_{n}(t,T)$ acts
on $x$ and $\hat{\Lc}^{\xi}_{n}(t,T)$ acts on $\xi$.  With a slight abuse of terminology, we call
$\hat{\Lc}^{\xi}_{n}$ the \emph{symbol}\footnote{The operator $\hat{\Lc}^{\xi}_{n}$ is not a
function as in the classical theory of pseudo-differential calculus.  However $\ee^{-\ii
\<\xi,x\>} \hat{\Lc}^{\xi}_{n} \ee^{\ii \<\xi,x\>}$ is the symbol of $\Lc_n^x(t,T)$.} of the
operator $\Lc_{n}^x(t,T)$ in \eqref{eq:Ln}.
%As we shall prove below, $\hat{\Lc}^{\xi}_{n}$ is a \emph{differential operator} and not a function as in the classical theory of pseudo-differential calculus: this is due to the fact that $\Lc^{x}_{n}(t,T)$ is an operator with variable (polynomial) coefficients.

Let us consider the operator $\Mc^{x}(t,t_k)\equiv \Mc(t,t_k)$ in \eqref{eq:M.def} and denote by
$\Mc^{x}_{i}(t,t_k)$ its $i$th component. The symbol $\widehat{\Mc}_{i}^{\xi}(t,t_k)$ of
$\Mc^{x}_{i}(t,t_k)$ is defined analogously to \eqref{e33}, that is
\begin{align}\label{e35}
  \Mc^{x}_{i}(t,t_k)\ee^{\ii \<\xi, x \>}=\widehat{\Mc}_{i}^{\xi}(t,t_k)\ee^{\ii \<\xi, x \>}.
\end{align}
Explicitly, we have
\begin{align}
 \widehat{\Mc}_{i}^{\xi}(t,t_k)= F_{i}(\xi,t,t_k)-\ii \partial_{\xi_{i}},\qquad i=1,\dots,d,
\end{align}
where the function $F$ is defined as%\ora{What is $\widehat{\Mc}_{i}^{\xi}(t,t_k)$?}
\begin{align}
F_{i}(\xi,t,t_k)=\int_{\Rb^d} \int_{t}^{t_k} z_{i} \( \ee^{\ii\< z,\xi \>} - 1 \)
            \nu_0(s,\dd z)\dd s
            + \int_{t}^{t_k}  m_{i}(s)\dd s + \ii\int_{t}^{t_k}  \left(C(s)\xi\right)_{i} \dd s.
\end{align}
We note that, while $\Mc^{x}$ is a first order \emph{integro-differential} operator, its symbol
$\widehat{\Mc}^{\xi}$ is a first order \emph{differential} operator. For this reason, it is more
convenient to use the symbol $\widehat{\Mc}^{\xi}$ instead of the operator $\Mc^{x}$. Note also
that
\begin{align}
\Mc^{x}_{i}(t,t_k)\Mc^{x}_{j}(t,t_k)\ee^{\ii \<\xi, x \>}=\Mc^{x}_{i}(t,t_k)\widehat{\Mc}_{j}^{\xi}(t,t_k)\ee^{\ii \<\xi, x \>}
  =\widehat{\Mc}_{j}^{\xi}(t,t_k)\Mc^{x}_{i}(t,t_k)\ee^{\ii \<\xi, x \>}=\widehat{\Mc}_{j}^{\xi}(t,t_k)\widehat{\Mc}_{i}^{\xi}(t,t_k)\ee^{\ii \<\xi, x \>}.
\end{align}
Since $\Mc^{x}_{i}$ and $\Mc^{x}_{j}$ commute when applied to a function that admits a Fourier
representation, then $\widehat{\Mc}_{j}^{\xi}$ and $\widehat{\Mc}_{i}^{\xi}$ also commute when
applied to such functions. In particular, the operator
$\left(\widehat{\Mc}^{\xi}(t,t_{k})\right)^{\beta}$, for $\beta\in\mathbb{N}^{d}_{0}$, is well
defined and we have
\begin{align}\label{e40}
\left(\widehat{\Mc}^{\xi}(t,t_{k})\right)^{\beta}\ee^{\ii \<\xi, x \>}=
\left(\Mc(t,t_k)\right)^{\beta}\ee^{\ii \<\xi, x \>}.
\end{align}
%\ora{Again, I do not see why this is the case, because I do not yet understand why the componentes of $\Mc$ commute.}
%We remark that the order in which the operators
%are composed both in $\left(\Mc(t,t_k)\right)^{\beta}$ and
%$\left[\widehat{\Mc}^{\xi}(t,t_{k})\right]^{\beta}$ is relevant.
From identity \eqref{e40} we obtain directly the expression of the symbol of $\Gc_{j}$ in \eqref{eq:G.def}.  Indeed, recalling the expression \eqref{e34} of $\nu_{j}$ we have
\begin{align}
\hat{\Gc}^{\xi}_j(t,t_k)
    %&= \Ac_j\left(t_k,\widehat{\Mc}^{\xi}(t,t_k)\right) \red{\quad \textrm{this inequality is not correct, is it? The terms do not commute..}}\\
        &= \sum_{|\beta|\le M_{j}} \int_{\Rb^d}\left(\ee^{ \ii\<z, \xi \>} - 1 - \ii\<z, \xi \>\right)\nu_{j,\beta}\left(t_k,\dd z\right)\, \left(\widehat{\Mc}^{\xi}(t,t_k)\right)^{\b}
            + \sum_{|\alpha |\leq 2} \left(\ii\xi\right)^{\a}
            a_{\alpha,j}\left(t_k,\widehat{\Mc}^{\xi}(t,t_k)\right).
    \label{eq:Gh.def}
\end{align}
Thus we have proved the following lemma
 %that gives the explicit expression of
%$\hat{\Lc}_{n}(t,T)$.
\begin{lemma}
\label{lemmand}
We have
\begin{align}
 \hat{\Lc}^{\xi}_{n}(t,T)
    &=  \sum_{k=1}^n
            \int_{t}^T \dd t_1 \int_{t_1}^T \dd t_2 \cdots \int_{t_{k-1}}^T \dd t_k
            \sum_{i \in I_{n,k}}
            \hat{\Gc}^{\xi}_{i_1}(t,t_1)
            \hat{\Gc}^{\xi}_{i_{2}}(t,t_{2})
            \cdots \hat{\Gc}^{\xi}_{i_k}(t,t_k)  , \label{eq:Lnh}
\end{align}
where $I_{n,k}$ as defined in \eqref{eq:Ink}.
\end{lemma}
%%{\bf\blue [Andrea: make a Lemma here for formula \eqref{e33}-\eqref{eq:Lnh}?]}
%\ora{I think we need to be slightly more explicit here.  Please see the step-by-step computations
%I provided in the proof of Theorem \ref{thm:dyson}.  I think this is the level of detail we need
%to provide.}
\noindent The following theorem extends the Fourier pricing formula \eqref{eq:u0} to higher order
approximations.
%  {\blue\bf [Andrea: I would state the following theorem in terms of characteristic functions and NOT of solutions. It would be more clear and easy to apply to compute option prices via Fourier.]}
\begin{theorem}
\label{thm:fourier}
Under the assumptions of Proposition \ref{thm:dyson},
%Let $\phi$ belong to the Schwartz space.  Then,
for any $n\geq 1$ we have %$\uh_n(t,\xi)$, the Fourier transform of $u_n(t,x)$, is given by
\begin{align}
  u_n(t)&=\frac{1}{(2 \pi )^d}  \int_{\Rb^d} \hat{P}_n(t,x,T,\xi)\phih(-\xi)\,\dd \xi, \label{eq:un_fourier}
   %&=\frac{1}{(2 \pi )^d}  \int_{\Rb^d}
            %\ee^{\Phi_0(t,T,\xi)}\Lc_n(t,T)e^{\ii \< \xi, x \>} \phih(-\xi)\,\dd \xi \\
\end{align}
where $\hat{P}_n(t,x,T,\xi)$ is the $n$th order term of the approximation of the characteristic
function of $X$ (cf. Remark \ref{r5}). Explicitly, we have
\begin{align}\label{e22}
 \hat{P}_n(t,x,T,\xi):=\hat{P}_0(t,x,T,\xi)
    \(\ee^{-\ii \<\xi,x \>} \hat{\Lc}^{\xi}_{n}(t,T)\ee^{\ii \<\xi, x \>} \)
\end{align}
where $\hat{P}_0(t,x,T,\xi)$ is the $0$th order approximation in \eqref{eq:P0} and
$\hat{\Lc}^{\xi}_{n}(t,T)$ is the differential operator defined in \eqref{eq:Lnh}.
%
%$\chi_n(t,T,x,\xi):=\frac{}{e^{\ii \< \xi, x \>}}$ being the symbol of the operator $\Lc_n(t,T)$.
\end{theorem}

\begin{proof}
We first note that, since the approximating operator $\Lc^{x}_n$ acts in the $x$ variables, then
it commutes\footnote{This was one of the main points of the {\it adjoint expansion method}
proposed by \cite{pascucci}.} with the Fourier pricing operator \eqref{eq:u0}. Thus, by
\eqref{eq:un} combined with \eqref{eq:u0}, we get
\begin{align}
u_n(t)=\Lc^{x}_n(t,T)u_0(t)&=\frac{1}{(2 \pi )^d}  \int_{\Rb^d}\Lc^{x}_n(t,T)
            \ee^{\ii \< \xi, x \> + \Phi_0(t,T,\xi)} \phih(-\xi)\, \dd \xi\\
            &=\frac{1}{(2 \pi )^d}  \int_{\Rb^d}
            \hat{P}_0(t,x,T,\xi)
                        \(\ee^{-\ii \<\xi,x \>} \hat{\Lc}^{\xi}_{n}(t,T)\ee^{\ii \<\xi, x \>} \)
                        \phih(-\xi)\, \dd
            \xi,
%            \intertext{($\Lc_n(t,T)$ only acts on the $x$ variable)}
%         & =\frac{1}{(2 \pi )^d}  \int_{\Rb^d}
%            \ee^{ \Phi_0(t,T,\xi)} \phih(-\xi)\Lc_n(t,T)\ee^{\ii \< \xi, x \> }\, \dd \xi  .
\end{align}
and the thesis follows from \eqref{e33}. %since $\Lc_n(t,T)$ only acts on the $x$-variables, thus commuting with the
%Fourier transform operator.
\end{proof}
\begin{remark}
%Note that the symbol $\chi_n$ in the Fourier representation \eqref{eq:un_fourier} for the $u_n$ is
%actually fully explicit.
Computing the term in parenthesis above $\(\ee^{-\ii \<\xi,x \>} \hat{\Lc}^{\xi}_{n}(t,T)\ee^{\ii
\<\xi, x \>} \)$ is a straightforward exercise since the symbol $\hat{\Lc}^{\xi}_{n}(t,T)$, given
in \eqref{eq:Lnh}, is a differential operator.
% As a matter of example, we report here its
%expression for $n=1$ under the Taylor expansion for the coefficients proposed in
%Example \ref{example:Taylor}: % (i.e. $B_1(x)=x-\bar{x}$):
%\begin{align}
%  \frac{\Lc_1(t,T)e^{\ii \<\xi, x \>}}{e^{\ii \<\xi,x \>}}  =\sum_{j=1}^d  \int_t^T \bar{\Ac}^j_1(s,\xi)\left(  x_j -\bar{x}_j +
%\int_{t}^{s}   m_{j}(r)\dd r +\ii\sum_{k=1}^d \xi_k  \int_{t}^{s} C_{j,k}(r) \dd r - \ii
%\partial_{\xi_j}\Psi_{0}(t,s,\xi)
%   \right) ds  ,
%   \end{align}
%with $m,C$ and $\Psi_{0}$ as in \eqref{eq:m}-\eqref{eq:C}-\eqref{eq:Psi0},
%%  $$\Psi_{0}(t,T,\xi)=\int_{t}^{T} \int_{\Rb_0^d}\nu_0(s,\dd z)\dd s\, \( \ee^{\ii \< \xi, z \>} - 1 - \ii \< \xi, z \> \),$$
%and where
%\begin{align}
%%\Ac_n(t,x)&=
%\bar{\Ac}^j_1(s,\xi)= \sum_{|\alpha |\leq 2}  \partial_{x_j} a_{\alpha}(s,\bar{x}) (\ii\xi)^{\alpha} +\int_{\Rb_0^d} \left( \ee^{\ii \<\xi,z \>} - 1 - \ii\<\xi,z \>\right) \partial_{x_j} h_{\alpha}(s,\bar{x}) \bar{\nu}(s,\dd z).
%\end{align}
\end{remark}

\begin{remark}
In case of non-integrable payoffs (e.g. Call and Put options), the Fourier representation
\eqref{eq:un_fourier} can be easily extended by considering the Fourier transform on the imaginary
line $\xi=\xi_r + \ii \xi_{\ii} $. For instance, since the Call option payoff $\phi(x)=\left(\ee^x
-\ee^k\right)^+$ is not integrable, its Fourier transform $\phih(-\xi)$ must be computed in a
generalized sense by fixing an imaginary component of the Fourier variable $\xi_{\ii}<-1$.
%Thus, the price of a European call option can be computed using standard Fourier methods
%\begin{align}
%u(t,x,y)
%    &=  \frac{1}{2\pi} \int_\Rb \dd \xi_r \, \phih(-\xi) \Eb_{x,y} \ee^{ \ii \xi X_{T-t}} , &
%\phih(\xi)
%    &=  \frac{-\ee^{k-\ii k \xi}}{ \ii \xi + \xi^2 } , &
%\xi
%    &=  \xi_r + \ii \xi_i , &
%\xi_i
%    &<  -1 . \label{eq:u.Heston}
%\end{align}
%Note that,
\end{remark}

\begin{remark}
Observe that the $N$th order approximation \eqref{eq:ubar.N}-\eqref{eq:un_fourier} requires only a
single Fourier inversion
\begin{align}\label{eq:fourier_price_N}
\ub_N(t,x)
    &=  \sum_{n=0}^N u_n(t,x)
        =       \frac{1}{(2\pi)^d}\sum_{n=0}^N \int_{\Rb^d}  \hat{P}_n(t,x,T,\xi)\phih(-\xi)% \frac{\Lc_n(t,T)e^{\ii \<\xi, x \>}}{e^{\ii \<\xi,x \>}}
        \,  \dd \xi .
\end{align}
Moreover, when evaluating the inverse transform, the number of dimensions over which one must
integrate numerically is equal to the number of components of $x$ that appear in the option payoff
$\varphi$. This is due to the fact that the Fourier transform of a constant is a Dirac delta
function. In particular, let $\varphi(x)\equiv \bar{\varphi}(\bar{x})$ with
$\bar{x}=(x_1,\cdots,x_{d'})$, for some $d'<d$. Then we have
$\hat{\varphi}(\xi)=(2\pi)^{d-d'}\hat{\bar{\varphi}}\left(\bar{\xi}\,\right)\delta_0({\xi_{d'+1}})\cdots\delta_0({\xi_{d}})$
with $\bar{\xi}=(\xi_1,\cdots,\xi_{d'})$, and thus
\begin{align}
\ub_N(t,x)
    &        =       \frac{1}{(2\pi)^{d'}}\sum_{n=0}^N \int_{\Rb^{d'}} \hat{P}_n\left(t,x,T,\left(\bar{\xi},0\right)\right) \hat{\bar{\varphi}}\left(-\bar{\xi}\,\right) \,\dd \bar{\xi}.
\end{align}
%Thus, $\phih$ will contain a $\del(\xi_i)$ if $x_i$ does not appear in $\phi$.  Note that derivatives of %Dirac delta functions are defined in the distributional sense.  That is
%\begin{align}
%\int \dd \xi_i \, \hat{g}(\xi_i) \d_{\xi_i}^n \del(\xi_i)
%    &= \int \dd \xi_i \,  \del(\xi_i) (-1)^n \d_{\xi_i}^n \hat{g}(\xi_i)
%    =       (-1)^n \d_{\xi_i}^n \hat{g}(0) .
%\end{align}
\end{remark}

\section{Example: Heston model with stochastic jump-intensity}
\label{sec:heston}
Consider the following model for an asset $S = \ee^X$, written under the pricing measure $\Qb$ assuming zero interest rates
\begin{align}
\dd X_t
    &=  \( -\frac{1}{2} - \int_\Rb \nu(\dd \zeta)( \ee^{\zeta}-1-\zeta) \) Z_t \dd t + \sqrt{Z_t} \dd W_t
            + \int_\Rb \zeta \dd \Nt(t,Z_t,\dd t,\dd \zeta) ,% &\Eb \dd N_t(z,\dd \zeta)
    %&=  z \nu ( \dd \zeta ) \dd t ,
        \\
\dd Z_t
    &=  \kappa (\theta - Z_t ) \dd t + \del \sqrt{Z_t} \dd B_t , \qquad
\dd \< W, B \>_t
    =  \rho \dd t .
\end{align}
Note that, just as in the Heston model, the instantaneous volatility of $X$ is given by
$\sqrt{Z_t}$, where $Z$ is a CIR process.  Likewise, the instantaneous arrival rate of jumps of
size $\dd \zeta$ is given by $Z_t \nu(\dd \zeta)$, where $\nu$ is a L\'evy measure satisfying all
of the usual integrability conditions.  The generator $\Ac$ of the process $(X,Z)$ is given by
\begin{align}
\Ac
    &=    z \( \mu \d_x + \frac{1}{2} \d_x^2 + \int_\Rb \nu(\dd \zeta) ( \ee^{\zeta \d_x} - 1 - \zeta \d_x) \)
                + \kappa (\theta - z) \d_z + \frac{1}{2} \del^2 z \d_z^2 + \rho \del z \d_x \d_y , \\
\mu
    &=  - \frac{1}{2} - \int_\Rb \nu(\dd \zeta) ( \ee^{\zeta} - 1 - \zeta) .
\end{align}
%\ora{Andrea and Stefano: I do not remember changing the reference from \cite{carr2004time} to \cite{bates1996jumps}.  The model considered above is {not} the Bates model.  The Bates model has a constant jump intensity whereas the jump intensity considered in \cite{carr2004time} is proportional to a CIR process, which is the case for the model we consider.  So, I am changing the citation back to \cite{carr2004time}.}
%\par
%\ora{I am always grateful for the work both of you do on our papers.  Your ideas are extremely valuable and you are able to make things much more rigorous than I ever could.  But, if you completely remove something I have written, I would really appreciate it if you tell me about this before you remove it.  That way, I will understand the proposed changes better and, if there is a disagreement, we can discuss it.}\\
%{\blue Stefano: You are right. But I am sure that whoever changed it had just forgotten to underly the change.}
%%As in \cite{bates1996jumps},
The characteristic function $\hat{P}(t,x,z,T,\xi) := \Eb[ \ee^{\ii \xi X_T}| X_t = x, Z_t = z]$ is obtained in \cite{carr2004time} by expressing the process $X$ as a time-changed L\'evy process.  One can also obtain the characteristic function by solving for the Fourier transform of the fundamental solution corresponding to the operator $(\d_t + \Ac)$.  We have
\begin{align}
\hat{P}(t,x,z,T,\xi)
    &=  \ee^{\ii \xi x + C(T-t,\xi) + z\,D(T-t,\xi)} , \\
C(\tau,\xi)
    &=  \frac{\kappa \theta}{ \del^2} \( (\kappa - \rho \delta  \ii \xi + d(\xi) ) \tau
            -2 \log \[ \frac{1-f(\xi) \ee^{d(\xi)\tau}}{1-f(\xi)}\]\) , \\
D(\tau,\xi)
    &=  \frac{\kappa - \rho \del \ii \xi + d(\xi)}{\del^2} \frac{1-\ee^{d(\xi)\tau}}{1-f(\xi) \ee^{d(\xi)\tau}} , \\
f(\xi)
    &=  \frac{\kappa - \rho \del \ii \xi + d(\xi)}{\kappa - \rho \del \ii \xi - d(\xi)} , \\
d(\xi)
    &=  \sqrt{ -\del^2 \, 2 \psi(\xi) + (\kappa - \rho \ii \xi \del)^2} , \\
\psi(\xi)
     &=      \ii \mu \xi -\tfrac{1}{2} \xi^2 + \int_\Rb \nu(\dd \zeta)(\ee^{\ii \xi \zeta} - 1 - \ii \xi \zeta) .
\end{align}
With an explicit expression for $\hat{P}(t,x,z,T,\xi)$ available, the price of a European call option can be computed using standard Fourier methods
\begin{align}
u(t,x,z)
    &=  \frac{1}{2\pi} \int_\Rb \dd \xi_r \, \hat{P}(t,x,z,T,\xi) \phih(-\xi) , &
\phih(\xi)
    &=  \frac{-\ee^{k-\ii k \xi}}{ \ii \xi + \xi^2 } , &
\xi
    &=  \xi_r + \ii \xi_i , &
\xi_i
    &<  -1 . \label{eq:u.Heston}
\end{align}
Note that, since the call option payoff $\phi(x)=(\ee^x -\ee^k)^+$ is not in $L^1(\Rb)$, its Fourier
transform $\phih(\xi)$ must be computed in a generalized sense by fixing an imaginary component of
the Fourier variable $\xi_i < -1$.
\par
Also of interest are sensitivities of option prices or \emph{Greeks}.  In particular, consider the $\Delta$ and the $\Gam$, which are defined as
\begin{align}
\Delta(t,x,z)
    &:= \d_s u(t,x(s),z)
    =       \ee^{-x} \d_x u(t,x,z) , \label{eq:Delta} \\
\Gam(t,x,z)
    &:= \d_s^2 u(t,x(s),z)
    =       \ee^{-2x}( \d_x^2 - \d_x )u(t,x,z) , \label{eq:Gamma}
\end{align}
where we have used $x(s)=\log s$.
When computing terms of the form $\d_x^m u(t,x,z)$, observe that the differential operator $\d_x^m$ acts only on the characteristic function $\hat{P}$ appearing in \eqref{eq:u.Heston} and not on the Fourier transform $\hat{\phi}$ of the payoff ${\phi}$.  Likewise, when using Theorem \ref{thm:fourier} to compute $\d_x^m \ub_n(t,x,z) = \sum_{i=0}^n \d_x^m u_i(t,x,z)$ the differential operator $\d_x^m$ acts only on $\hat{P}_i$ in \eqref{eq:un_fourier}.
%Note that, %by the Fourier representation \eqref{eq:ift} we have
%%\begin{align}
%%D_x^\alpha f(x)
%%    &=  \frac{1}{(2 \pi )^d}\int_{\Rb^d} \dd \xi \, \fh(\xi)
%%                D_x^\alpha \ee^{ -\ii \< \xi , x\> }
%%        =   \frac{1}{(2 \pi )^d}\int_{\Rb^d} \dd \xi \, (-\ii \xi )^\alpha \fh(\xi)
%%                 \ee^{ -\ii \< \xi , x\> } . \label{eq:greeks}
%%\end{align}
%by the standard identity $\Fc[D^\alpha f](\xi)=(-\ii \xi )^\alpha \hat{f}(\xi)$ and Theorem
%\ref{thm:fourier}, one can quickly compute the $\Delta$ and the $\Gamma$ in \eqref{eq:Delta} and
%\eqref{eq:Gamma}.
\par
Now, we specialize to the case where jumps are normally distributed
\begin{align}
\nu(\dd \zeta)
    &=  \frac{\lam}{\sqrt{2 \pi s^2}}\exp \( \frac{-(\zeta - m)^2}{2 s^2} \) .
\end{align}
In Figure \ref{fig:HestonJumps} we plot the implied volatility $\sig$ corresponding to the exact
price $u$ as well as the implied volatility $\sigb_2$ corresponding to our second order
approximation $\ub_2$.  To compute $\sig$ we first compute option prices using \eqref{eq:u.Heston}; we then invert the Black-Scholes equation numerically in order to obtain the implied volatility $\sig$.  To compute our second order approximation of implied volatility $\sigb_2$ we first compute our second order approximation for prices $\ub_2$ using Theorem \ref{thm:fourier}; we then invert the Black-Scholes equation numerically in order to obtain $\sigb_2$.  Values from Figure \ref{fig:HestonJumps} can be found in Table
\ref{tab:IV}.  In Figure \ref{fig:Delta} we plot the exact $\Del$ as well as our second order approximation $\bar{\Del}_2$.  In Figure \ref{fig:Gamma} we plot the exact $\Gamma$ as well as our second order approximation $\bar{\Gam}_2$.  Values from Figures \ref{fig:Delta} and
\ref{fig:Gamma} are given in Tables \ref{tab:Delta} and \ref{tab:Gamma} respectively.  Exact Greeks are computed by combining \eqref{eq:u.Heston}, \eqref{eq:Delta} and \eqref{eq:Gamma}.  Approximate Greeks are computed by combining Theorem \ref{thm:fourier} and equations \eqref{eq:Delta} and \eqref{eq:Gamma}.

%%%%%%%%%%%%%%%%%%%%%%%%%%%%%%%%%%%%%
%
%       Conclusion
%
%%%%%%%%%%%%%%%%%%%%%%%%%%%%%%%%%%%%%

\section{Conclusion}
In this paper we derive a family of asymptotic expansions for European option prices when the underlying is modeled as a $d$-dimensional time inhomogeneous L\'evy-type process.  By combining the classical Dyson series expansion with a novel polynomial expansion of the generator, we obtain two equivalent representations for approximate option price: (i) as an integro-differential operator acting on the order zero price, and (ii) as a Fourier transform.  We implement our pricing approximation on a Heston-like model which allows for both stochastic volatility and stochastic jump intensity.  We find that our second order expansion provides and excellent approximation for prices (as seen through corresponding implied volatilities), as well as for the Greeks $\Del$ and $\Gam$.

%%%%%%%%%%%%%%%%%%%%%%%%%%%%%%%%%%%%%%%%%%%%%%%%%%%%
%
%           Bibliography
%
%%%%%%%%%%%%%%%%%%%%%%%%%%%%%%%%%%%%%%%%%%%%%%%%%%%%

\bibliographystyle{chicago}
\bibliography{Bibtex-Master-3.00}

%%%%%%%%%%%%%%%%%%%%%%%%%%%%%%%%%%%%%%%%%%%%%%%%%%%%
%
%           Figures
%
%%%%%%%%%%%%%%%%%%%%%%%%%%%%%%%%%%%%%%%%%%%%%%%%%%%%

\begin{figure}
\centering
\begin{tabular}{ | c | c | }
\hline
$t=0.10$ & $t=0.25$ \\
\includegraphics[width=.46\textwidth,height=0.23\textheight]{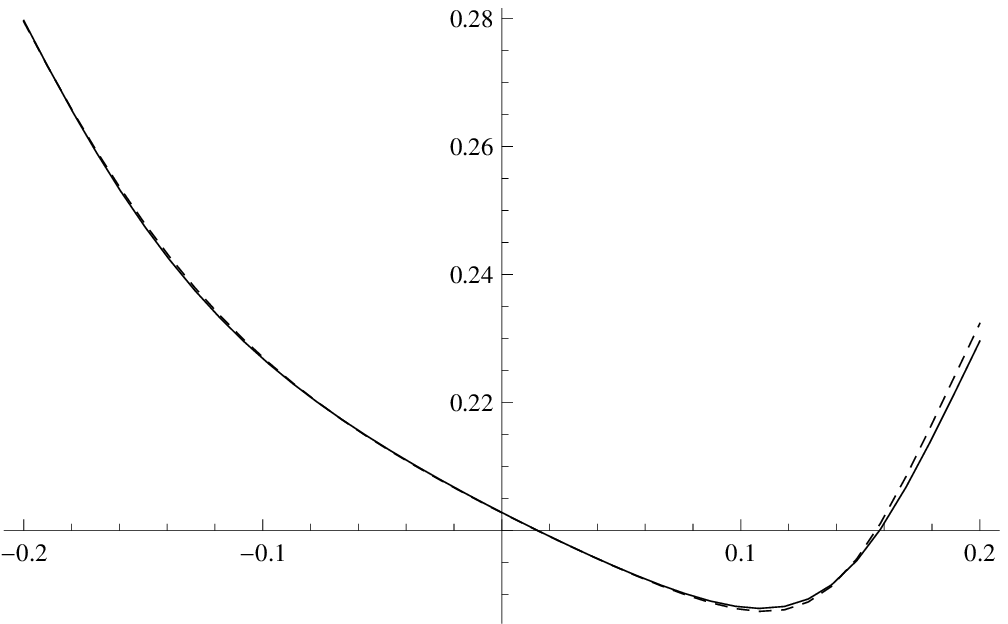} &
\includegraphics[width=.46\textwidth,height=0.23\textheight]{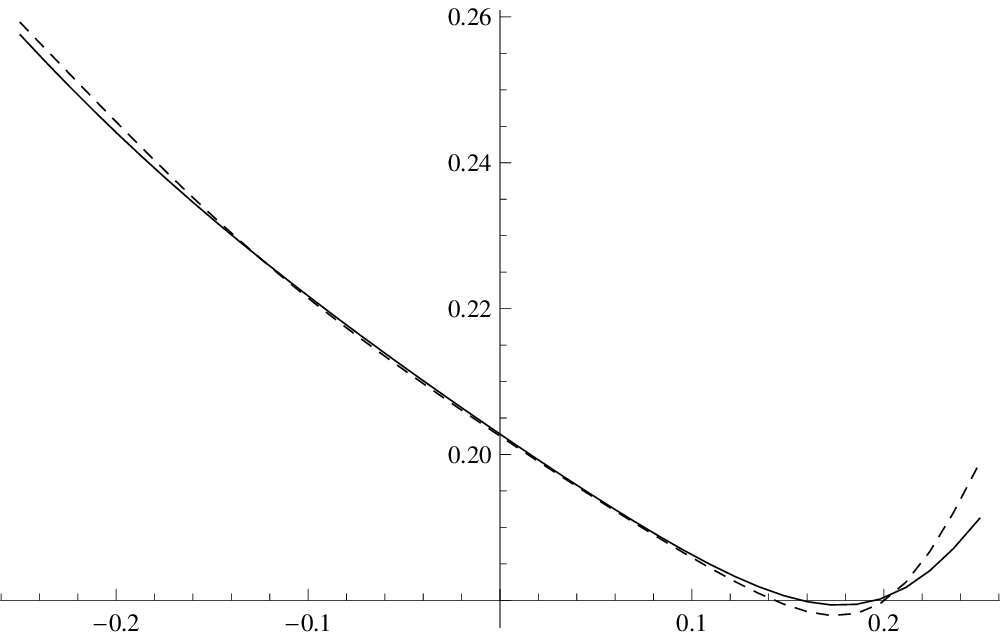}\\ \hline
$t=0.50$ & $t=1.00$ \\
\includegraphics[width=.46\textwidth,height=0.23\textheight]{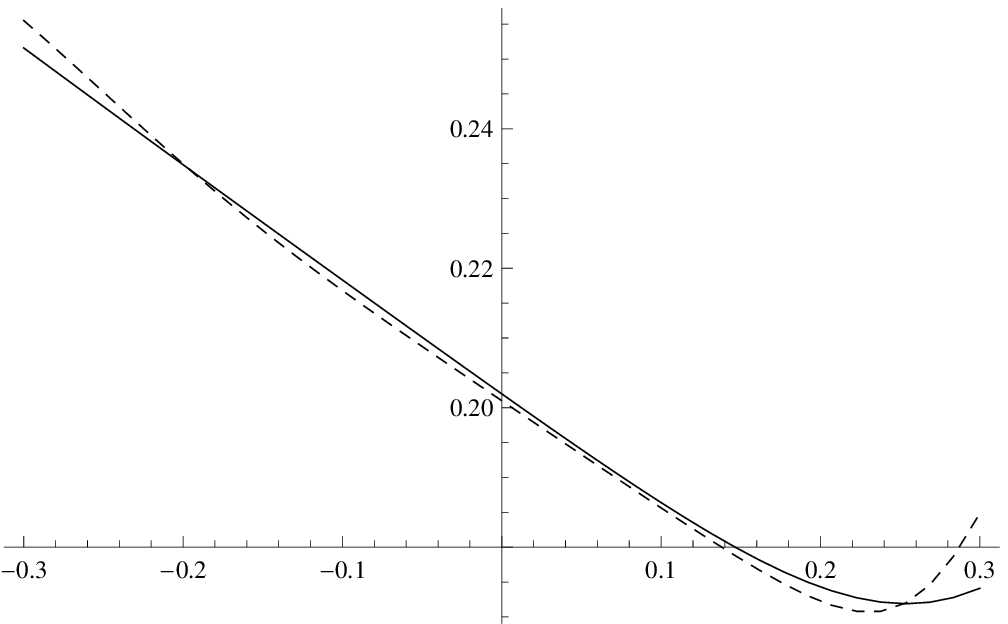} &
\includegraphics[width=.46\textwidth,height=0.23\textheight]{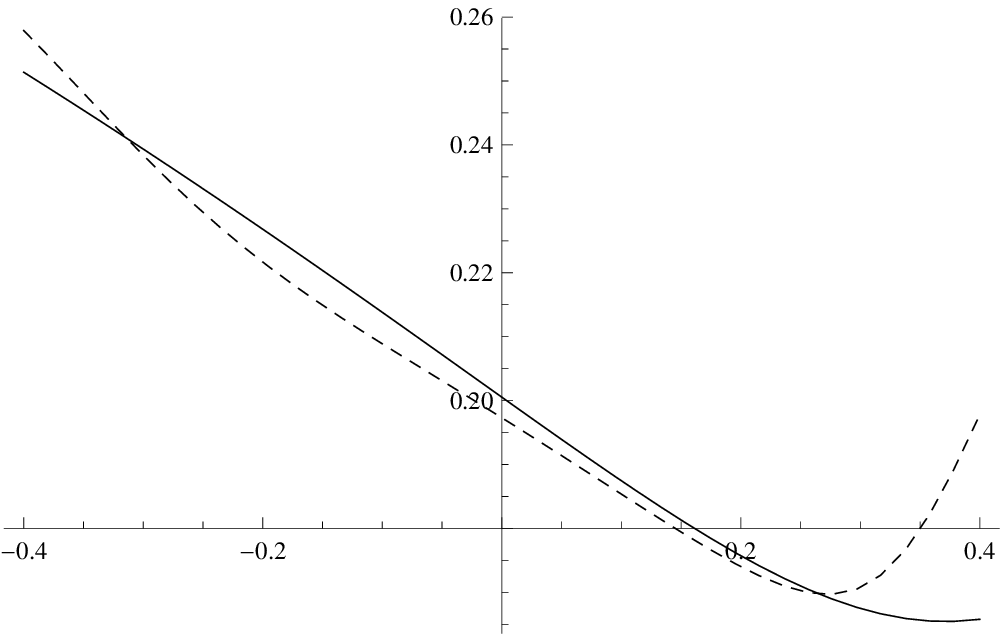}\\ \hline
\end{tabular}
\begin{align}
\nu(\dd \zeta)
    &=  \frac{\lam}{\sqrt{2 \pi s^2}}\exp \( \frac{-(\zeta - m)^2}{2 s^2} \) .
\end{align}
\caption{For the model considered in Section \ref{sec:heston}, we plot the implied volatility $\sig$ corresponding to the exact option price $u$ (solid black) as well as the implied volatility $\sigb_2$ corresponding to our second order option price approximation $\ub_2$ (dashed black).  The units of the horizontal axis are $\log$ strike $k:=\log K$.  Approximate prices are computed using the Taylor series expansion of $\Ac(t)$ as described in Example \ref{example:Taylor}.  We assume the L\'evy measure $\nu$ is as parametrized above.  The following parameters are used in all four plots: $\kappa = 1.15$, $\theta = 0.04$, $\del=0.2$, $\rho = -0.7$, $z = \theta$, $x=0$, $m=-0.1$, $s=0.2$, $\lam=2.0$.}
\label{fig:HestonJumps}
\end{figure}
% General-Expression-1.07

\begin{table}
\centering
\begin{tabular}{c|c|ccccccccc}
\hline
            & $k-x$                                 &  -0.2   & -0.15   &   -0.1  &     -0.05  &    0.00     &  0.05    &   0.1     &   0.15    &   0.2  \\
\hline
                & $\sig$                            &  0.2797   &  0.2478   &  0.2269   &  0.2133   &  0.2028    &  0.1940  &  0.1881    &  0.1960  &  0.2296\\
t=0.10  & $\bar{\sig}_2$            &  0.2795   &  0.2483   &  0.2271   &  0.2132   &  0.2028    &  0.1939  &  0.1877    &  0.1963  &  0.2324\\
                &   $\text{rel. err.}$  &  0.0006   &  0.0018   &  0.0009   &  0.0003   &  0.0002    &  0.0001  &  0.0020    &  0.0018  &  0.0120\\
\hline
            & $\sig$                            &  0.2441   &  0.2323    &  0.2217   &  0.2120   &  0.2028   &  0.1941   &  0.1863   &  0.1805   &  0.1803\\
 t=0.25 & $\bar{\sig}_2$            &  0.2456   &  0.2328    &  0.2215   &  0.2116   &  0.2025   &  0.1939   &  0.1859   &  0.1793   &  0.1799\\
                &   $\text{rel. err.}$  &  0.0059   &  0.0018    &  0.0013   &  0.0020   &  0.0013   &  0.0009   &  0.0021   &  0.0067   &  0.0027\\
\hline
            & $\sig$                            &  0.2348   &  0.2266    &  0.2183   &  0.2101   &  0.202    &  0.1940.  &  0.1864   &  0.1796   &  0.1743\\
 t=0.50 & $\bar{\sig}_2$            &  0.2350   &  0.2254    &  0.2168   &  0.2088   &  0.201    &  0.1933   &  0.1856   &  0.1783   &  0.1723\\
            &   $\text{rel. err.}$  &  0.0005   &  0.0049    &  0.0069   &  0.0063   &  0.004    &  0.0037   &  0.0040   &  0.0070   &  0.0116\\
\hline
            & $\sig$                            &  0.2268   &  0.2204    &  0.2138   &  0.2072   &  0.2005   &  0.1939   &  0.1875   &  0.1813   &  0.1757\\
t=1.00  & $\bar{\sig}_2$            &  0.2217   &  0.2149    &  0.2089   &  0.2031   &  0.1973   &  0.1914   &  0.1854   &  0.1794   &  0.1740\\
                &   $\text{rel. err.}$  &  0.0227   &  0.0246    &  0.0230   &  0.0197   &  0.0160   &  0.0130   &  0.0111   &  0.0103   &  0.0096\\
\hline
\end{tabular}
\caption{Exact implied vols $\sig$, second order approximation $\bar{\sig}_2$ and relative error $|(\bar{\sig}_2-\sig)/\sig|$.  Parameters are the same as those in Figure \ref{fig:HestonJumps}.}
\label{tab:IV}
\end{table}
% General-Expression-1.08

\begin{figure}
\centering
\begin{tabular}{ | c | c | }
\hline
$t=0.10$ & $t=0.25$ \\
\includegraphics[width=.46\textwidth,height=0.23\textheight]{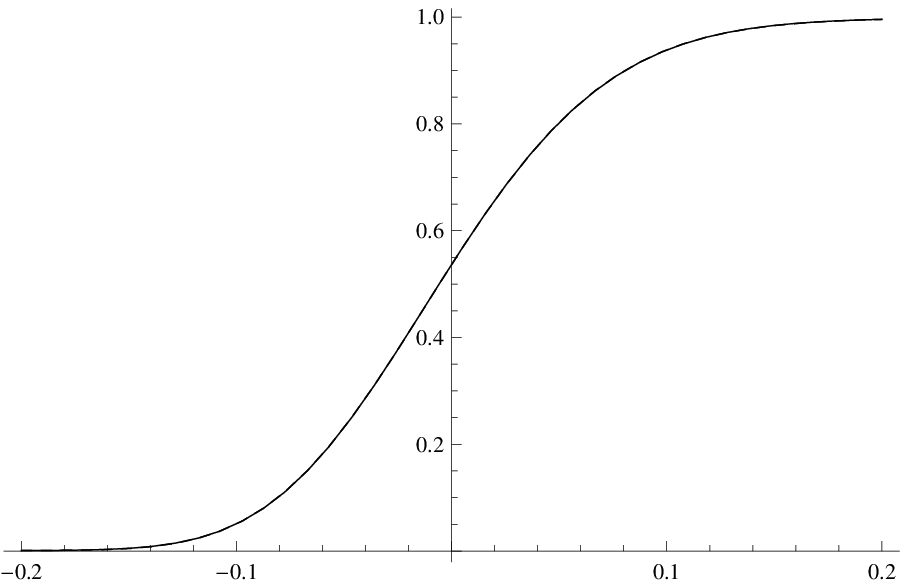} &
\includegraphics[width=.46\textwidth,height=0.23\textheight]{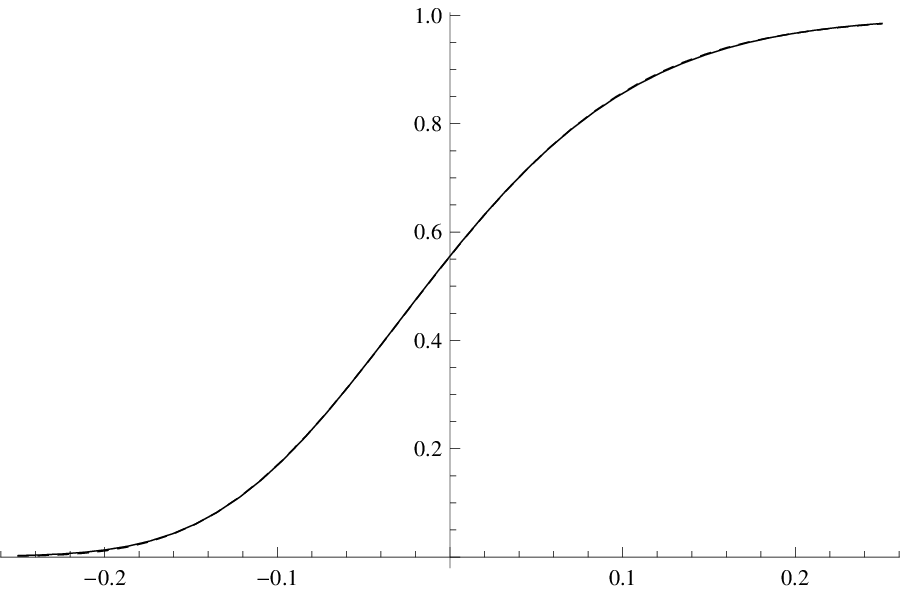}\\ \hline
$t=0.50$ & $t=1.00$ \\
\includegraphics[width=.46\textwidth,height=0.23\textheight]{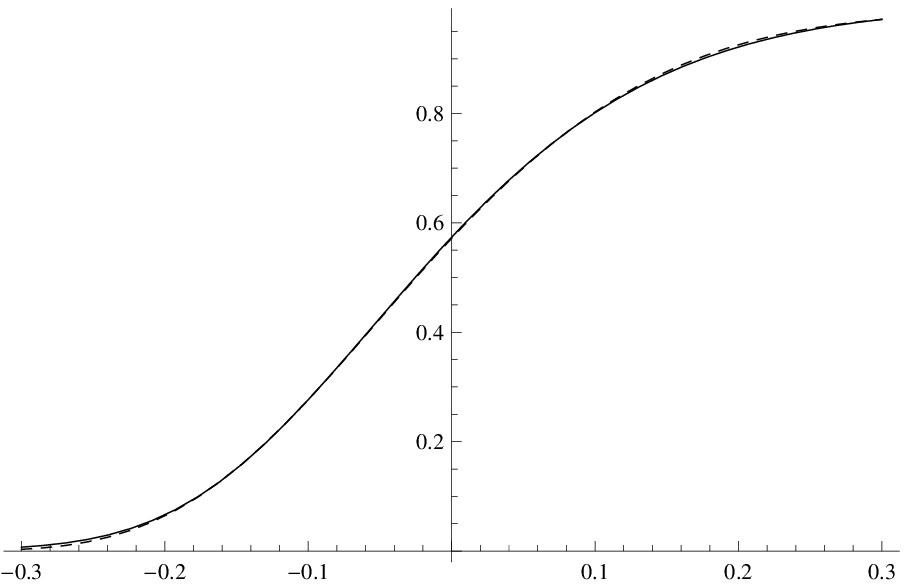} &
\includegraphics[width=.46\textwidth,height=0.23\textheight]{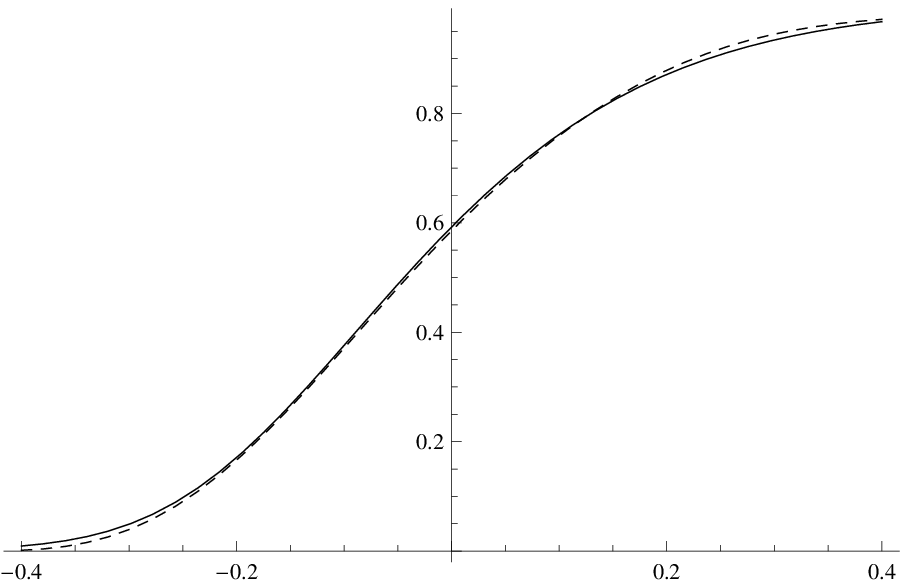}\\ \hline
\end{tabular}
\caption{For the model considered in Section \ref{sec:heston}, we plot the Delta $\Delta$ corresponding to the exact option price $u$ (solid black) as well as the Delta $\bar{\Delta}_2$ corresponding to our second order option price approximation $\ub_2$ (dashed black).  The units of the horizontal axis are $x$.  Approximate prices are computed using the Taylor series expansion of $\Ac(t)$ as described in Example \ref{example:Taylor}.  We assume the L\'evy measure $\nu$ is as given in Figure \ref{fig:HestonJumps}.  The following parameters are used in all four plots: $\kappa = 1.15$, $\theta = 0.04$, $\del=0.2$, $\rho = -0.7$, $z = \theta$, $k=0$, $m=-0.1$, $s=0.2$, $\lam=2.0$.}
\label{fig:Delta}
\end{figure}
% General-Expression-1.07

\begin{table}
\centering
\begin{tabular}{c|c|ccccccccc}
\hline
            & $x$                                &  -0.2   & -0.15   &  -0.1  &     -0.05  &    0.00     &  0.05    &   0.1     &   0.15    &   0.2  \\
\hline
                & $\Del$                             &  0.0008   &  0.00516  &  0.05084  &  0.2312   &  0.5370   &  0.8024   &  0.9385   &  0.9845   &  0.9959  \\
t=0.10  & $\bar{\Del}_2$             &  0.0009   &  0.00478  &  0.05081  &  0.2313   &  0.5368   &  0.8026   &  0.9387   &  0.9843   &  0.9958  \\
                &   $\text{rel. err.}$   &  0.1309   &  0.07358  &  0.00048  &  0.0006   &  0.0003   &  0.0002   &  0.0002   &  0.0002   &  0.0000  \\
\hline
            & $\Del$                             &  0.01311  &  0.05708  &  0.1690   &  0.3503   &  0.5559   &  0.7329   &  0.8563   &  0.9293   &  0.9672  \\
 t=0.25 & $\bar{\Del}_2$             &  0.0114   &  0.05674  &  0.1696   &  0.3502   &  0.5552   &  0.7330   &  0.8576   &  0.9306   &  0.9673  \\
                &   $\text{rel. err.}$   &  0.1305   &  0.00585  &  0.0035   &  0.0004   &  0.0012   &  0.0000   &  0.0014   &  0.0014   &  0.0000  \\
\hline
            & $\Del$                             &  0.06608  &  0.1506   &  0.2767   &  0.4260   &  0.5739   &  0.7018   &  0.8014   &  0.8731   &  0.9215  \\
 t=0.50 & $\bar{\Del}_2$             &  0.06425  &  0.1508   &  0.2766   &  0.4246   &  0.5719   &  0.7007   &  0.8027   &  0.8766   &  0.9256  \\
            &   $\text{rel. err.}$   &  0.02773  &  0.0014   &  0.0003   &  0.0032   &  0.0034   &  0.0015   &  0.0015   &  0.0040   &  0.0044  \\
\hline
            & $\Del$                             &  0.1708   &  0.2667   &  0.3760   &  0.4878   &  0.5927   &  0.6849   &  0.7618   &  0.8234   &  0.8713  \\
t=1.00  & $\bar{\Del}_2$             &  0.1662   &  0.2627   &  0.3710   &  0.4814   &  0.5857   &  0.6791   &  0.7595   &  0.8262   &  0.8789  \\
                &   $\text{rel. err.}$   &  0.0268   &  0.01496  &  0.0131   &  0.0130   &  0.0117   &  0.0084   &  0.0030   &  0.0033   &  0.0088  \\
\hline
\end{tabular}
\caption{Exact Delta $\Del$, second order approximation $\bar{\Del}_2$ and relative error $|(\bar{\Del}_2-\Del)/\Del|$.  Parameters are the same as those in Figure \ref{fig:Delta}.}
\label{tab:Delta}
\end{table}
% General-Expression-1.08

\begin{figure}
\centering
\begin{tabular}{ | c | c | }
\hline
$t=0.10$ & $t=0.25$ \\
\includegraphics[width=.46\textwidth,height=0.23\textheight]{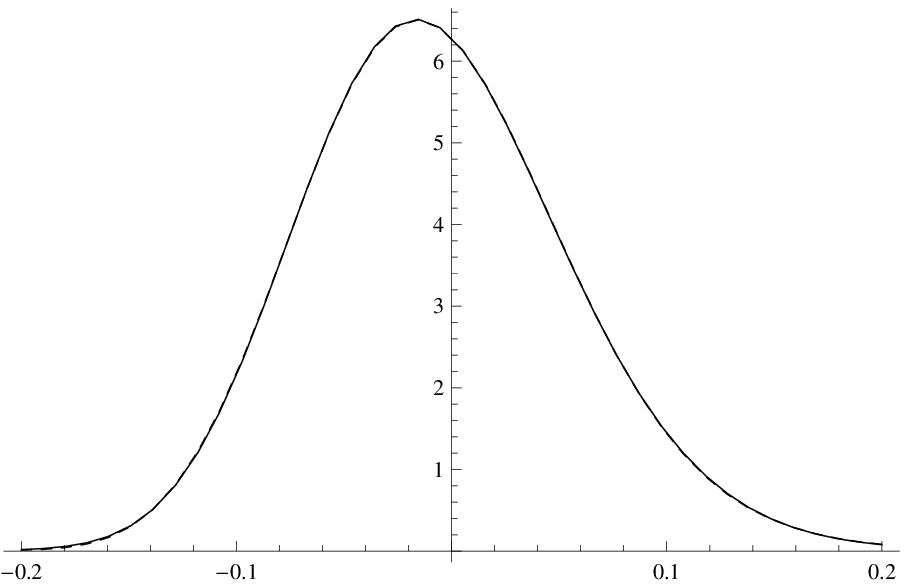} &
\includegraphics[width=.46\textwidth,height=0.23\textheight]{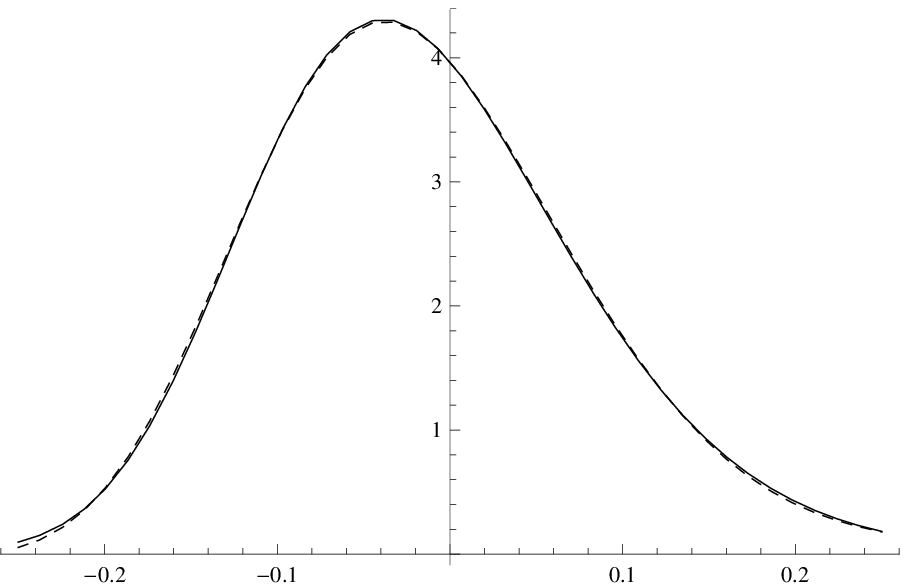}\\ \hline
$t=0.50$ & $t=1.00$ \\
\includegraphics[width=.46\textwidth,height=0.23\textheight]{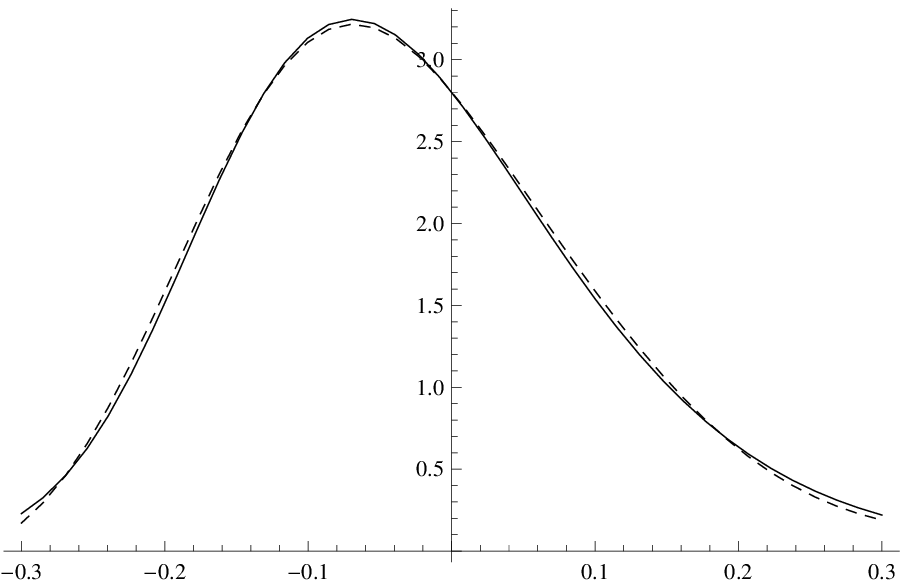} &
\includegraphics[width=.46\textwidth,height=0.23\textheight]{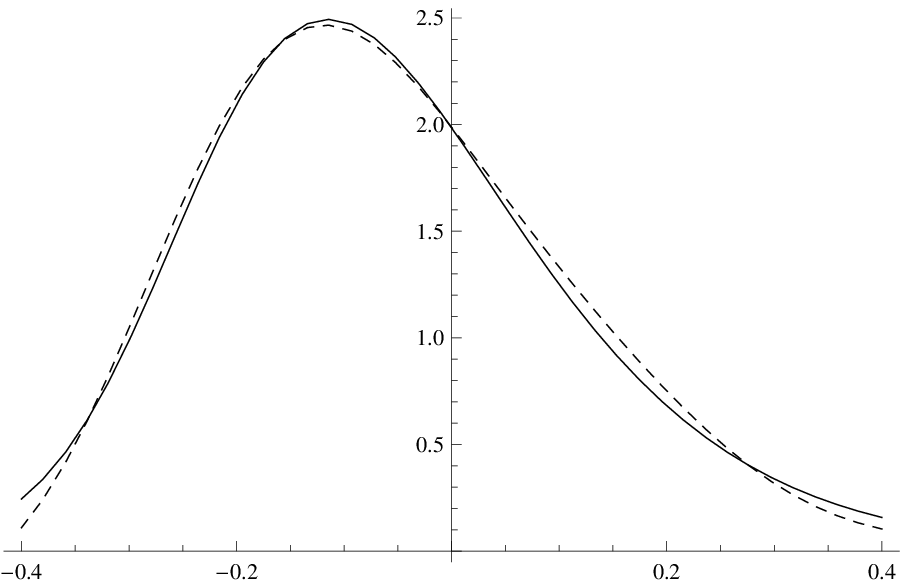}\\ \hline
\end{tabular}
\caption{For the model considered in Section \ref{sec:heston}, we plot the Gamma $\Gamma$ corresponding to the exact option price $u$ (solid black) as well as the Gamma $\bar{\Gamma}_2$ corresponding to our second order option price approximation $\ub_2$ (dashed black).  The units of the horizontal axis are $x$.  Approximate prices are computed using the Taylor series expansion of $\Ac(t)$ as described in Example \ref{example:Taylor}.  We assume the L\'evy measure $\nu$ is as given in Figure \ref{fig:HestonJumps}.  The following parameters are used in all four plots: $\kappa = 1.15$, $\theta = 0.04$, $\del=0.2$, $\rho = -0.7$, $z = \theta$, $k=0$, $m=-0.1$, $s=0.2$, $\lam=2.0$.}
\label{fig:Gamma}
\end{figure}
% General-Expression-1.07

\begin{table}
\centering
\begin{tabular}{c|c|ccccccccc}
\hline
            & $x$                               &  -0.2   & -0.15   &   -0.1  &     -0.05  &    0.00     &  0.05    &   0.1     &   0.15    &   0.2  \\
\hline
                & $\Gam$                             &  0.01828  &  0.2978   &  2.159    &  5.539    &  6.288    &  3.831    &  1.446    &  0.3779   &  0.0780  \\
t=0.10  & $\bar{\Gam}_2$             &  0.01197  &  0.2897   &  2.1760   &  5.5300   &  6.288    &  3.841    &  1.437    &  0.3748   &  0.0821  \\
                &   $\text{rel. err.}$   &  0.3452   &  0.0273   &  0.0077   &  0.0015   &  0.0001   &  0.0025   &  0.0061   &  0.0082   &  0.0518  \\
\hline
            & $\Gam$                             &  0.5185   &  1.705    &  3.337    &  4.275    &  3.967    &  2.884    &  1.738    &  0.906    &  0.4229  \\
 t=0.25 & $\bar{\Gam}_2$             &  0.5267   &  1.747    &  3.334    &  4.255    &  3.969    &  2.907    &  1.754    &  0.8925   &  0.4016  \\
                &   $\text{rel. err.}$   &  0.0157   &  0.024    &  0.0009 &  0.0046     &  0.0003 &  0.0079     &  0.0094   &  0.0149   &  0.0503  \\
\hline
            & $\Gam$                             &  1.514    &  2.488    &  3.135    &  3.206    &  2.802    &  2.174    &  1.54     &  1.017    &  0.635  \\
 t=0.50 & $\bar{\Gam}_2$             &  1.585    &  2.508    &  3.109    &  3.182    &  2.804    &  2.208    &  1.588    &  1.045    &  0.6244  \\
            &   $\text{rel. err.}$   &  0.0468   &  0.0079   &  0.0081   &  0.0076   &  0.0007   &  0.015    &  0.0309   &  0.0279   &  0.0167  \\
\hline
            & $\Gam$                             &  2.095    &  2.425    &  2.483    &  2.306    &  1.985    &  1.612    &  1.251    &  0.9364   &  0.6814  \\
t=1.00  & $\bar{\Gam}_2$             &  2.134    &  2.418    &  2.452    &  2.280    &  1.988    &  1.656    &  1.331    &  1.028    &  0.7511  \\
                &   $\text{rel. err.}$   &  0.0183   &  0.0032   &  0.0124   &  0.0110   &  0.0015   &  0.0276   &  0.0644   &  0.097    &  0.1023  \\
\hline
\end{tabular}
\caption{Exact Gamma $\Gam$, second order approximation $\bar{\Gam}_2$ and relative error $|(\bar{\Gam}_2-\Gam)/\Gam|$.  Parameters are the same as those in Figure \ref{fig:Gamma}.}
\label{tab:Gamma}
\end{table}
% General-Expression-1.08

\end{document}